\documentclass[a4paper,11pt]{article}
\usepackage{fullpage}
\usepackage{hyperref}

\usepackage[utf8]{inputenc}
\usepackage{amssymb,amsmath,amsfonts,amsthm}  
\usepackage{xspace,enumerate}
\usepackage[dvipsnames]{xcolor}
\usepackage{thmtools}
\usepackage[cmbtt]{bold-extra}
\usepackage[T1]{fontenc}
\usepackage{makecell}
\usepackage{listings}
\usepackage{multirow}
\usepackage{diagbox}
\usepackage{tabularx}
\usepackage{epsfig}
\usepackage{verbatim}
\usepackage[noadjust]{cite}
\usepackage{subcaption}
\usepackage{tikz}
\usetikzlibrary{positioning,arrows,arrows.meta,calc,patterns,patterns.meta,fit}

\usepackage[capitalise]{cleveref}
\let\Cref\cref
 
\newtheorem{theorem}{Theorem}[section]
\newtheorem{@theorem}{Theorem}[section]
\newtheorem{fact}{Fact}[section]
\newtheorem{corollary}{Corollary}[section]
\newtheorem{lemma}{Lemma}[section]

\newtheorem{definition}{Definition}[section]
\newtheorem{claim}{Claim}[section]

\crefname{conjecture}{Conjecture}{Conjectures}
\crefname{lemma}{Lemma}{Lemmas}
\crefname{problem}{Problem}{Problems}
\crefname{remark}{Remark}{Remarks}
\crefname{definition}{Definition}{Definitions}
\crefname{observation}{Observation}{Observations}
\crefname{@theorem}{Theorem}{Theorems}
\crefname{fact}{Fact}{Facts}
\crefname{claim}{Claim}{Claims} 

\def\dd{\mathinner{.\,.}}
\newcommand{\R}{\mathbb{R}}
\newcommand{\Z}{\mathbb{Z}}

\renewcommand{\O}{\mathcal{O}}
\newcommand{\tO}{\tilde{\mathcal{O}}}

\newcommand{\Q}{\mathcal{Q}}

\renewcommand{\S}{\mathcal{S}}

\newcommand{\U}{\mathcal{U}}
\newcommand{\V}{\mathcal{V}}

\renewcommand{\L}{\mathcal{L}}
\newcommand{\W}{\mathcal{W}}
\newcommand{\I}{\mathcal{I}}

\newcommand{\Ta}{T_\mathbf{a}}

\renewcommand{\phi}{\varphi}
\newcommand{\floor}[1]{\left\lfloor #1 \right\rfloor}
\newcommand{\set}[1]{\left\lbrace #1 \right\rbrace}
\newcommand{\bigset}[1]{\big \lbrace #1 \big \rbrace}

\newcommand{\eq}[1]{\begin{align*} #1 \end{align*}}

\DeclareMathOperator*{\Edges}{E}
\DeclareMathOperator*{\X}{X}
\DeclareMathOperator*{\Y}{Y}
\DeclareMathOperator*{\score}{score}
\DeclareMathOperator*{\Ham}{Ham}
\DeclareMathOperator*{\ID}{Id}
\DeclareMathOperator*{\dom}{dom}
\DeclareMathOperator*{\chrome}{C}

\newcommand{\absolute}[1]{\left\lvert#1\right\rvert}

\newcommand{\defproblem}[3]{
\vspace{2mm}
\noindent\fbox{
  \begin{minipage}{0.94\textwidth}
    \textsc{\large #1}\\
    {\bf{Input:}} #2  \\
    {\bf{Output:}} #3
  \end{minipage}
  }
\vspace{2mm}
}

\graphicspath{{drawings/}}

\begin{document}
\title{Faster two-dimensional pattern matching with $k$ mismatches}

\author{Jonas Ellert\thanks{DI/ENS, PSL Research University, France. Partially supported by the grant ANR-20-CE48-0001 from the French National Research Agency (ANR).}
\and Paweł Gawrychowski\thanks{Institute of Computer Science, University of Wrocław, Poland. Partially supported by the Polish National Science Centre grant number 2023/51/B/ST6/01505.}
\and Adam Górkiewicz\thanks{Institute of Computer Science, University of Wrocław, Poland. Partially supported by the Polish National Science Centre grant number 2023/51/B/ST6/01505.}
\and Tatiana Starikovskaya\thanks{DI/ENS, PSL Research University, France. Partially supported by the grant ANR-20-CE48-0001 from the French National Research Agency (ANR).}
}



\date{}
\maketitle

\begin{abstract}
The classical pattern matching asks for locating all occurrences of one string, called the pattern, in another, called the text, where
a string is simply a sequence of characters.
Due to the potential practical applications, it is desirable to seek approximate occurrences, for example by bounding the number
of mismatches. This problem has been extensively studied, and by now we have a good
understanding of the best possible time complexity as a function of $n$ (length of the text), $m$ (length of the pattern), and $k$
(number of mismatches).
In particular, we know that for $k=\O(\sqrt{m})$, we can achieve quasi-linear time complexity [Gawrychowski and Uznański, ICALP 2018].

We consider a natural generalisation of the approximate pattern matching problem to two-dimensional strings, which are simply
square arrays of characters. The exact version of this problem has been extensively studied in the early 90s, motivated by the
potential applications in image processing. While periodicity, which is the basic tool for one-dimensional pattern matching, 
admits a natural extension to two dimensions, it turns out to become significantly more challenging to work with,
and it took some time until an alphabet-independent linear-time algorithm has been obtained by Galil and Park [SICOMP 1996].

In the approximate two-dimensional pattern matching, we are given a pattern of size $m\times m$ and a text of size $n\times n$,
and ask for all locations in the text where the pattern matches with at most $k$ mismatches.
The asymptotically fastest algorithm for this algorithm works in $\O(kn^{2})$ time [Amir and Landau, TCS 1991].
We provide a new insight
into two-dimensional periodicity to improve on these 30-years old bounds.
Our algorithm works in $\tO((m^{2}+mk^{5/4})n^{2}/m^{2})$ time, which is $\tO(n^{2})$ for $k=\O(m^{4/5})$.
\end{abstract}

\setcounter{page}{0}
\thispagestyle{empty}
\clearpage

\section{Introduction}
Strings are basic objects that can used to store and manipulate sequential data, readily available in any popular
programming language. The fundamental algorithmic problem on such objects is pattern matching: identifying occurrences of one string in another.
Traditionally, the former is called the pattern and the latter the text.
An asymptotically optimal linear-time algorithm for this problem is known since the 70s~\cite{Knuth1977}, and multiple dozens
efficient algorithms have been described in the literature since then~\cite{DBLP:books/daglib/0025563}. However, 
from the point of view of possible applications, in particular those in computational biology, where we work with biological sequences
that can be described by finite strings over the alphabet $\{A,C,G,T\}$, a more appropriate task is to search for approximate occurrences, that is, to 
allow for some slack when defining when two such strings are equal.

From a theoretical point of view, a particular clean notion of an approximate occurrence is that of bounded Hamming distance: Given an integer parameter $k$ together with a text of length $n$ and a pattern of length $m$, the task is to find all
positions in the text where the pattern occurs with at most $k$ mismatches. This is known as pattern matching with $k$
mismatches. The natural assumption is that $k$ is not too large,
and the running time should be close to linear when $k$ is small. Indeed, from the point of view of potential applications we do
not want positions where there are many mismatches, as they would not be very meaningful.

Pattern matching with $k$ mismatches was already studied in the 80s, when a technique informally known as
the ``kangaroo jumping'' was introduced~\cite{Landau1986,Galil1986}. In this approach,
each position in the text is considered one-by-one, and we calculate the number of mismatches
by jumping over regions where there is no mismatch, moving to the next position when it exceeds $k$.
A single jump can be implemented in constant time with a data structure for the longest common extensions, such
as a suffix tree augmented with a lowest common ancestors structure (or some of the earlier simpler structures with slightly worse
parameters). This results in an $\O(nk)$ time algorithm, which is efficient for small values of $k$.
Another early discovery, efficient when there are very few distinct characters present in the input,
is to run the fast Fourier transform for each distinct character, resulting in $\O(n|\Sigma|\log n)$ running time~\cite{FischerP74}.
While, in principle, this is possibly quite bad for a general input, in the late 80s Abrahamson showed
how to apply the so-called heavy-light trick to guarantee that the time complexity is always
$\O(n \sqrt{m \log m})$ time~\cite{Abrahamson1987}. This clearly established that the $\O(nk)$ time complexity
can't be the right answer for the whole range values of $k$.

It was only in 2004 that both bounds were unified to obtain an $\O(n\sqrt{k \log k})$ time algorithm~\cite{Amir2004},
by combining the fast Fourier transform with a purely combinatorial reasoning that allows for eliminating all but a few relevant positions
in the text. It might have been seen at the time that such a complexity is the best possible as a function of $n$, $m$, and $k$,
but surprisingly in 2016 Clifford et al.~\cite{Clifford2016a} designed another algorithm that works
in $\tO(n + k^2n/m)$ time\footnote{We write $\tO$ to hide factors polylogarithmic in $n$.}. In particular, this is almost linear when $k=\O(\sqrt{m})$. Later, their
approach was refined to work in $\tO(n + kn/\sqrt{m})$ time~\cite{Gawrychowski2018},
which gives a smooth trade-off between $\tO(n\sqrt{k})$ and $\tO(n + k^2n/m)$. 
It is known that a significantly faster algorithm implies fast Boolean matrix multiplication~\cite{Gawrychowski2018},
and the time complexity can be slightly improved to $\O(n + kn\sqrt{(\log m) / m})$ \cite{Chan2020} (at the expense
of allowing Monte Carlo randomization). In a very recent exciting improvement, it was shown how to slightly improve
these time complexities by leveraging a connection to the 3-SUM problem~\cite{Chan0WX23}.
While further improvements might be still possible, it is certainly the case that one-dimensional pattern matching with
bounded Hamming distance is a fairly well-understood problem.
This is also the case from the more combinatorial point of view: we know that occurrences of the pattern with $k$ mismatches
either have a simple and exploitable structure, or the pattern is close to being periodic~\cite{Bringmann2019,Charalampopoulos2020a}.

\paragraph{2D strings.} The natural extension of strings to two dimensions are arrays of characters, called 2D strings. 
To avoid multiplying the parameters, we will assume that they are squares, but one could as well consider rectangles.
Naturally, such an extension is motivated by the possible application in image processing and the related areas,
and not surprisingly the corresponding problem of 2D exact pattern matching has been already considered in the 70s.
In this problem, the task is to locate all occurrences of an $m\times m$ pattern in an $n\times n$ text.
For constant-size alphabets, an asymptotically optimal $\O(n^{2}+m^{2})$-time algorithm has been found relatively
quickly~\cite{AMIR1992233,AMIR19922,BIRD1977168,doi:10.1137/0207043}, but (as opposed to 1D exact pattern matching) it was quite unclear for quite some time how to obtain
such a complexity without any assumptions on the size of the alphabet. Only in the mid 90s, a systematic
study of the so-called periodicities in 2D strings paved up the way for obtaining such an algorithm~\cite{Amir1994,Galil1996}
(even in logarithmic space~\cite{Crochemore1995}). Efficient parallel algorithms have also been obtained~\cite{Crochemore1998},
and the time complexity for random inputs, i.e., average time complexity, has also been considered~\cite{Baeza-Yates1993,Tarhio1996,Kaerkkaeinen1999}.

\paragraph{Periodicities in 2D strings.} The fundamental combinatorial tool used for 1D strings is periodicity,
defined as follows. We say that $p$ is a period of a string $S[1 \dd n]$ when $S[i]=S[i+p]$, for all $i$ such that the expression is defined.
The set of all periods of a given string is very structured~\cite{Fine1965}, namely, when $p$ and $q$ are both periods
of the same string $S[1 \dd n]$, and further $p+q\leq n$, then so is $\gcd(p,q)$.
For 2D strings, the notion of periodicity becomes more involved. In particular, instead of periodic
and non-periodic strings it is now necessary to consider four types of strings: non-periodic, lattice periodic, line periodic, and
radiant periodic~\cite{Amir1998}. Such considerations formed the foundations of the asymptotically optimal alphabet-independent
2D exact pattern matching algorithms~\cite{Amir1994,Galil1996}.
Later, some purely combinatorial properties of two-dimensional periodicities have been discovered~\cite{Mignosi2003,Gamard2017},
but generally speaking repetitions in two-dimensional strings are inherently more complicated than in one-dimensional strings.
For example, compressed pattern matching for two-dimensional strings becomes NP-complete~\cite{Berman2002}, see~\cite{Rytter2000}
for a more extensive discussion.
Another example, perhaps less extreme, is the bounds on two-dimensional runs~\cite{Amir2020} and
distinct squares~\cite{Charalampopoulos2020}, where we know that increasing the dimension incurs at least an additional
logarithmic factor~\cite{Charalampopoulos2020}.

\paragraph{2D pattern matching with \boldmath$k$\unboldmath{} mismatches.} The next step for 2D pattern matching is to allow $k$ mismatches.
Already in 1987, an $\tO(kmn^{2})$ time algorithm was obtained for this problem~\cite{Krithivasan1987}. This was
soon improved to $\tO((k+m)n^2)$ time~\cite{Ranka1991}, and finally to $\O(kn^2)$~\cite{Amir1991}, which remains
to be the asymptotically fastest algorithm. A number of non-trivial results have been obtained under the assumption
that the input is random, i.e., for the average time complexity~\cite{Baeza-Yates1998,Park1998,Kaerkkaeinen1999}.
Given that other notions of approximate occurrences, e.g.~bounded edit distance, seem less natural in the two-dimensional
setting~\cite{Baeza-Yates1998a}, the natural challenge is to better understand the complexity
of 2D pattern matching with $k$ mismatches. Following the line of research for 1D pattern matching with $k$
mismatches, a particularly natural question is to understand if we are able to design a quasi-linear time
algorithm for polynomial $k=\O(n^{\epsilon})$ number of mismatches, where $0 < \epsilon$ is a small constant. It appears that no such result was known
in the literature, and it was not clear if the methods designed for 1D pattern matching with $k$ mismatches
can be immediately adapted to 2D strings, as they are based on extending the notion of periodicity, which becomes
inherently more complicated in two dimensions.

\paragraph{Our result.} We design an algorithm that, given an $n\times n$ text and $m\times m$ pattern, finds
all occurrences with at most $k$ mismatches of the former in the latter in $\tO((m^2 + mk^{5/4})n^2 / m^2)$ time.
This significantly improves on the previously known upper bound of $\O(kn^{2})$ (from over 30 years ago),
and provides a quasi-linear time algorithm for $k=\O(m^{4/5})$. To obtain our result, we follow the framework
used to solve one-dimensional pattern matching with bounded Hamming distance. However, due to the intrinsically
more complex nature of two-dimensional periodicities, this turns out to require a deeper delve into the geometric
properties of such objects.

\paragraph{Overview of the techniques.}
The starting point for our algorithm is the approach designed for the one-dimensional version, see e.g.~\cite{Gawrychowski2018}
for an optimized version (but the approach is due to~\cite{Clifford2016}), which proceeds as follows. Using the standard trick~\cite{Abrahamson1987}, we can assume w.l.o.g.~$n = 2m$. 
First, we approximate the Hamming distance for every position in the text with Karloff's algorithm~\cite{Karloff1993}.
Then, we can eliminate positions for which the approximated distance is very large.
If the number of remaining positions is small enough, we can use kangaroo jumps~\cite{Galil1986} to verify them one by one.
Otherwise, some two remaining possible occurrences must have a large overlap, and thus induce a small approximate period in the pattern, i.e., an integer $p$ such that aligning the pattern with itself at a distance $p$ incurs few mismatches.
Then, (for $n=2m$), we can restrict our attention to the middle part of the text with the same approximate period $p$.
Both the pattern and the middle part of the text compress very well under the simple RLE compression, if we rearrange their characters
by considering the positions modulo $p$. In other words, they can be both decomposed into few subsequences
of the form $i, i+p, i+2p, \ldots, i+\alpha p$ consisting of the same character.
By appropriately plugging in an efficient algorithm for approximate pattern
matching for RLE-compressed inputs, this allows us to obtain the desired time complexity.

In the two-dimensional case, there is no difficulty in adapting Karloff's algorithm or kangaroo jumps, which allows
us to focus on the case where there are two possible occurrences with a large overlap. Here, the two-dimensional
case significantly departs from the one-dimensional case in terms of technical complications. In 2D,
the natural definition of a period is not an integer but a pair of integers, i.e., a vector.
However, to obtain a compressed representation of
a 2D string with small approximate period we actually need two such periods $\phi$ and $\psi$.
Roughly speaking, we require that the parallelogram spanned by $\phi$ and $\psi$ has a small area.
We show that two vectors with the required properties exist with some geometric considerations and
applying the Dilworth's theorem.

Next, we show that with $\phi$ and $\psi$ in hand we can decompose the pattern
into nicely structured monochromatic pieces. First, we define a lattice as the set of points such that for each two of them their difference is
of the form $s\phi+t\psi$, for some $s,t\in \Z$. Second, a tile is a set of points in a parallelogram with sides
parallel to $\phi$ and $\psi$. Then, each piece is a lattice restricted to a truncated tile, where truncated
means that we only consider the points in a rectilinear rectangle. We show that the pattern
can be decomposed into $\O(k)$ such pieces.

The natural next step is to similarly decompose the text. As in the one-dimensional case, we
need to focus on the middle part of the text that admits the same approximate period. This turns out
to somewhat challenging in two dimensions, but nevertheless it is not very technical to obtain
a partition into $\O(mk)$ pieces. Then, we consider each piece of the pattern and each piece of the text,
and convolve them to calculate their contribution to the number of mismatches, obtaining an $\tO(m^{2}+mk^{2})$ time
algorithm for $n\leq \frac{3}{2}m$ (which can be guaranteed by covering the text into overlapping
squares).

To obtain our final result, we need to proceed more carefully when partitioning the text. We partition
the relevant positions of the text into its periphery and the rest, and observe that (by carefully choosing the parameters) we only need to convolve the periphery with a small portion of the pattern. The remaining non-peripheral
part of the text can be partitioned into fewer pieces, which allows us to obtain the following theorem.

\def\mainThmTitle{Informal version}
\def\mainThmContent{%
Given a two-dimensional $m \times m$ pattern string $P$ and a two-dimensional $n \times n$ text string $T$ with $m \le n$,
there is an algorithm that solves the $k$-mismatch problem in $\tO((m^2 + mk^{5/4})n^2 / m^2)$ time.%
}

\begin{restatable}[\mainThmTitle]{@theorem}{restateThmMain}\label{th:main}
\mainThmContent%
\end{restatable}

\noindent Throughout the paper, we assume the standard word-RAM model of computation with words of size $\Omega(\log n)$.

\section{Preliminaries}
\label{sec:preliminaries}
\newcommand{\x}[1]{#1.x}
\newcommand{\y}[1]{#1.y}
\newcommand{\h}[1]{\phi \times #1}
\newcommand{\s}[1]{\psi \times #1}

\textbf{Geometric notations.} For $n \in \Z^+$, denote $[n] = \set{0, \dots, n - 1}$. For $u \in \R^2$, denote its coordinates as $\x{u}, \y{u}$, i.e., $u = (\x{u}, \y{u})$. Furthermore, for $u, v \in \R^2$, denote 
\begin{alignat*}{4}
u &\ +\ &&      v = (\x{u}+\x{v}, \y{u} + \y{v}) & \qquad\qquad
u &\ -\ &&      v = (\x{u}-\x{v}, \y{u}-\y{v})\\
u &\ \;\cdot\ &&  v = \x{u} \cdot \x{v} + \y{u} \cdot \y{v} & 
u &\ \times\ && v = \x{u} \cdot \y{v} - \y{u} \cdot \x{v}
\end{alignat*}
Alternatively, $u \cdot v = \absolute{u}\absolute{v} \cos \alpha$ and $u \times v = \absolute{u}\absolute{v} \sin \alpha$, where $\alpha$ is the angle between $u$ and $v$.

\begin{definition}
	For a set $U \subseteq \R^2$ we denote
	\[ \X(U) = \set{\x{u} : u \in U}, \quad \Y(U) = \set{\y{u} : u \in U}.\]
\end{definition}


\newcommand{\wild}{\texttt{?}}
\noindent \textbf{String notations.} In this work, $\Sigma$ denotes an \emph{alphabet}, a finite set consisting of integers polynomially bounded in the size of the input strings. The elements of $\Sigma$ are called \emph{characters}. Additionally, we consider a special character denoted by $\wild$ that is not in $\Sigma$ and is called a \emph{wildcard}.

\newcommand{\getchar}[1]{\chrome(#1)}
\newcommand{\pto}{\mathrel{\ooalign{\hfil$\mapstochar\mkern5mu$\hfil\cr$\to$\cr}}}
\renewcommand{\d}[1]{\dom(#1)}
\newcommand{\f}[1]{#1^\mathbf{f}}

A \emph{one-dimensional string} $S$ is a function from a finite continuous subrange $[a,b]\subset \mathbb Z$ to $\Sigma \cup \{\wild\}$. If $a = 0$, we will sometimes simply write $S = S(a)S(a+1)\ldots S(b)$.\footnote{We chose this definition, which is slightly different from the standard one, as it conveniently generalises to two dimensions.} We define a \emph{two-dimensional string} $S$ as a partial function $\Z^2 \pto \Sigma$. This is more general than seemingly necessary due to the technical details of the algorithm. We store strings as lists of point-character pairs. For a string $S$, let $\d{S}$ be its domain. The \emph{size} of $S$ is the size of its domain.
The \emph{width} of a non-empty two-dimensional string $S$ is defined as $\max \X(\d{S}) - \min \X(\d{S}) + 1$ and the \emph{height} as $\max \Y(\d{S}) - \min \Y(\d{S}) + 1$. Finally, we say that $S$ is \emph{partitioned} into strings $R_0, R_1, \ldots, R_{\ell-1}$ if $\d{S} = \sqcup_{i \in [\ell]} \d{R_i}$ and $R_i(u) = S(u)$ for all $i \in [\ell]$ and $u \in \d{R_i}$.

\newcommand{\MI}{\textnormal{\textsc{MI}}\xspace}

\begin{definition}[Hamming distance]
	Two characters $a,b \in \Sigma$ \emph{match} if $a = b$. The wildcard matches itself and all characters in $\Sigma$. 
	For strings $S, R$ (either both one-dimensional or both two-dimensional), define
	\begin{align*}
		\MI(S, R)\ =\ &\set{u : u \in \d{S} \cap \d{R}, S(u), R(u) \text{ do not match}}\\
		\Ham(S,R)\ =\ &\absolute{\MI(S,R)}
	\end{align*}			 
	We call the elements of $\MI(S,R)$ the \emph{mismatches} between $S$ and $R$.
\end{definition}

\begin{definition}[Shifting]
	For a set $V \subseteq \Z$ (resp., $V \subseteq \Z^2$) and $u \in \Z$ (resp., $u \in \Z^2$), denote $V + u := \set{v + u : v \in V}$.
	For a one-dimensional (resp., two-dimensional) string $S$ and $u \in \Z$ (resp., $u \in \Z^2$), define $S + u$ to be a string $R$ such that
	$\d{R} = \d{S} + u$ and $R(v) = S(v - u)$ for $v \in \d{R}$.
	Intuitively, we shift the domain of the string while maintaining the character values corresponding to its elements.
\end{definition}

\defproblem{The text-to-pattern Hamming distances problem}{One-dimensional (resp., two-dimensional) strings $P, T$}{$\Ham(T, P+q)$ for all $q \in \Z$ (resp., $q \in \Z^2$) such that $\d{P + q} \subseteq \d{T}$}

\defproblem{The $k$-mismatch problem}{One-dimensional (resp., two-dimensional) strings $P, T$, an integer $k \in \Z$}{$\min\{k+1,\Ham(T, P+q)\}$ for all $q \in \Z$ (resp., $q \in \Z^2$) such that $\d{P + q} \subseteq \d{T}$}

\noindent Traditionally, $P$ is referred to as \emph{the pattern} and $T$ as \emph{the text}.

\subsection{Generalisations of one-dimensional algorithms.}
We start by generalising existing one-dimensional algorithms to two dimensions, which we will use as subroutines in the final algorithm for the $k$-mismatch problem. 

For a two-dimensional string $S$ of width $w_S$ and height $h_S$ define its \emph{linearisation} as a one-dimensional string $\bar{S}$ of size $w_S \cdot h_S$, where $\forall x \in \X(\d{S}), y \in \Y(\d{S})$:

$$
\bar{S}(x \cdot h_S + y) =
\begin{cases}
\wild & \text{ if } S(x,y) \text{ is not defined}\\
S(x,y) & \text{ otherwise}
\end{cases}
$$

Intuitively, this can be seen as inscribing $S$ into a rectangle, filling non-defined character values with wildcards, and writing the rectangle two-dimensional string column-by-column. 

\begin{theorem}\label{kangaroos}
Assume that $P, T$ are two-dimensional strings where $\d{P} = [m] \times [m]$ and $\d{T} = [n] \times [n]$, for $m \le n$. Given a set $Q \subseteq \Z^2$, there is an algorithm that computes $d_q = \Ham(P + q, T) $ for every $q \in Q$ in total time $\O(n^2 + \sum_{q \in Q} d_q)$.
\end{theorem}
\begin{proof}
For every $i \in [m]$, define a one-dimensional string $P_i = P(i,0) P(i,1) \ldots P(i,m-1)$. Analogously, for every $i \in [n], j \in [n-m+1]$, define a one-dimensional string $T_{i,j} = T(i,j) T(i,j+1) \ldots T(i,j+m-1)$. 

We assign in $\O(n^2)$ time an integer identifier $\ID$ to each $P_i$ and $T_{i,j}$, satisfying the property that the identifiers are equal if and only if the strings are equal, in a folklore way: First, we build the suffix tree~\cite{DBLP:conf/focs/Weiner73} for $\bar{T} \$ \bar{P}$ in $\O(n^2)$ time, where $\$ \notin \Sigma \cup \{\wild\}$. Next, we trim the suffix tree at depth $m$ by a depth-first traverse in $\O(n^2)$ time as well. By definition, every string $P_i$ and $T_{i,j}$ now corresponds to a leaf of the tree, and distinct strings correspond to distinct leaves. We can now identify a string with the corresponding leaf. 

\begin{fact}[{\cite{Galil1986}}]\label{fact:kangaroo1D}
Given a one-dimensional string $S: [\ell] \rightarrow \Sigma$. There is a data structure that can be built in $\O(\ell)$ time and for all pairs of strings $S_1 = S(i) S(i+1) \ldots S(i+\ell')$, $S_2 = S(j) S(j+1) \ldots S(j+\ell')$, where $0 \le i,j < \ell-\ell'$, allows computing $\MI(S_1+j-i, S_2)$ in time $\O(\Ham(S_1,S_2))$. 
\end{fact}

Let $\ID(T) = \ID(T_{0,0}) \ID(T_{0,1}) \ldots \ID(T_{0,n-m}) \ldots \ID(T_{n-1,0}) \ID(T_{n-1,1}) \ldots \ID(T_{n-1,n-m}))$ and $\ID(P) = \ID(P_0)\ID(P_1) \ldots \ID(P_{m-1})$. 
We build the data structure of \cref{fact:kangaroo1D} for two strings: $\ID(T) \$ \ID(P)$ and $\bar{T} \$ \bar{P}$ (note that neither of the strings contains a wildcard). In total, the construction takes $\O(\absolute{\ID(T) \$ \ID(P)} + \absolute{\bar{T} \$ \bar{P}}) = \O(n^2)$ time.

Consider now $q \in Q$. To compute $\Ham(T, P+q)$, note that 

\begin{align*}
\Ham(T, P+q) &= \sum_{i=0}^{m-1} \Ham(T_{\x{q}, \y{q}+i}, P_i)\\
&= \sum_{0 \le i \le m-1 \; : \; \ID(T_{\x{q}, \y{q}+i}) \neq \ID(P_i)} \Ham(T_{\x{q}, \y{q}+i}, P_i) \\
&= \sum_{y \in \MI(\ID(T_{\x{q}, \y{q}}) \ldots \ID(T_{\x{q}, \y{q}+m-1}), \ID(P))} \Ham(T_{\x{q}, y}, P_{y-\y{q}})\\
\end{align*}

It follows that we can first use the data structure of \cref{fact:kangaroo1D} for  $\ID(T) \$ \ID(P)$ to compute the set $\MI =  \MI(\ID(T_{\x{q}, \y{q}}) \ldots \ID(T_{\x{q}, \y{q}+m-1}, \ID(P))$ in $\O(\absolute{\MI}) = \O(d_q)$ time, and then for each $y \in \MI$, we use the data structure of  \cref{fact:kangaroo1D}  for $\bar{T} \$ \bar{P}$ to compute 
$\Ham(T_{\x{q}, y}, P_{y-\y{q}})$ in total $\O(d_q)$ time as well. Summing up the resulting values, we obtain $\Ham(T, P+q)$. The claim follows. 
\end{proof}

The following fact is folklore, and follows by applying, for every distinct character present in $P$ and $T$, a separate instance of the fast Fourier transform~\cite{FischerP74} to copies of the pattern and the text where all but the given character are replaced with a new character.

\begin{fact}[folklore]\label{fact:sigman1d}
Given one-dimensional strings $P$ of size $m$ and $T$ of size $n \ge m$ (both can contain wildcards), the text-to-pattern Hamming distances problem can be solved in $\O(\sigma \cdot n \log m)$ time, where $\sigma$ is the number of distinct characters present in $P$ and $T$.
\end{fact}

\begin{corollary}[of~\cite{Karloff1993}]\label{cor:approx1d}
Let $\varepsilon > 0$ be a constant. Given one-dimensional strings $P$ of size $m$ and $T$ of size $n \ge m$ (both can contain wildcards), there is a $(1+\varepsilon)$-approximation algorithm that solves the text-to-pattern Hamming distances problem in $\O((n/\varepsilon^2)  \log^3 m)$ time.
\end{corollary}
\begin{proof}
Karloff's algorithm~\cite{Karloff1993} solves the text-to-pattern Hamming distances problem for $P,T$ in $\O((n/\varepsilon^2)  \log^3 m)$ time assuming that neither $P$ nor $T$ contains a wildcard. We show how to modify it slightly to allow for wildcards. 

The idea of Karloff's algorithm is as follows. For each $1 \le i \le r = (\log m/ \varepsilon)^2$, the algorithm considers a mapping $\mu_i : \Sigma \rightarrow \{0,1\}$. It then solves the text-to-pattern Hamming distances problem for $\mu_i(T) := \mu_i(T(0))\mu_i(T(1)) \ldots \mu_i(T(n-1))$ and $\mu_i(P) := \mu_i(P(0))\mu_i(P(1)) \ldots \mu_i(P(m-1))$. For $1 \le j \le n-m$, define $d_i(j) = \Ham(\mu_i(T), \mu_i(P)+j-1)$. The algorithm returns the value $d(j) = \sum_i d_i(j)$. Karloff showed that $\Ham(T, P+j-1) \le d(j) \le (1+\varepsilon) \cdot \Ham(T, P+j-1)$.

We modify the mappings slightly. For $a \in \Sigma \cup \{\wild\}$, define
$$
\tilde{\mu}_i(a) = 
\begin{cases}
\mu_i(a) &\text{ if } a \in \Sigma\\
\wild & \text{ if } a = \wild.
\end{cases}
$$
We then construct $\tilde{\mu_i}(T)$ and $\tilde{\mu_i}(P)$ by applying the mapping to the characters of the strings similarly to above, and compute the values $\tilde{d}_i(j) = \Ham(\tilde{\mu_i}(T), \tilde{\mu_i}(P)+j-1)$ and $\tilde{d}(j) = \sum_i d_i(j)$ using \cref{fact:sigman1d} in $\O(r \cdot n \log m) = \O((n/\varepsilon^2) \log^3 m)$ time. It is not difficult to see that wildcards do not contribute to the distances neither between $T$ and $P$, nor between $\tilde{\mu_i}(T)$ and $\tilde{\mu_i}(P)$. Therefore, $\Ham(T, P+j-1) \le \tilde{d}(j) \le (1+\varepsilon) \cdot \Ham(T, P+j-1)$, and the algorithm is correct. 
\end{proof}

We now show a simple embedding from two-dimensional strings to one-dimensional, which is essential for generalising \cref{fact:sigman1d} and \cref{cor:approx1d} to two dimensions.

\begin{claim}
\label{claim:padding}
Let $P$ and $T$ be non-empty two-dimensional strings of respective widths $w_P, w_T$ and heights~$h_P, h_T$. There are one-dimensional strings $\lambda(P)$ of size $w_T \cdot h_P$ and $\lambda(T)$ of size $w_T \cdot h_T$ that we can compute in $\O((w_P+w_T)(h_P+h_T))$ time such that to solve the text-to-pattern Hamming distances problem for $P$ and $T$ it is enough to solve it for $\lambda(P)$ and $\lambda(T)$. 
\end{claim}
\begin{proof}
We first construct a two-dimensional string $P'$ such that for all $\min \X(\d{P}) \le i \le \max \X(\d{P})+(w_T-w_P)$ and $\min \Y(\d{P}) \le j \le \max \Y(\d{P})$ there is
$$
P'(i,j) = 
\begin{cases}
P(i,j) & \text{ if } (i,j) \in \d{P}\\
\wild & \text{ otherwise.}
\end{cases}
$$
Intuitively, we inscribe $P$ into a rectangle, then add $w_T-w_P$ columns to it, and finally fill all non-defined character values with wildcards. Note that for all $q \in \Z^2$ there is $\Ham(T,P+q) = \Ham(T,P'+q)$. We linearise $T$ to obtain a one-dimensional string $\bar{T}$ of size $w_T \cdot h_T$ and $P'$ to obtain a one-dimensional string $\bar{P'}$ of length $w_T \cdot h_P$. The strings $\bar{T}$ and $\bar{P'}$ can be computed in $\O((w_P+w_T)(h_P+h_T))$ time. Furthermore, we have:
$$\Ham(T, P+q) = \Ham(T,P'+q) = \Ham(\bar{T}(\x{q} \cdot h_T + \y{q}) \ldots \bar{T}(\x{q} \cdot h_T + \y{q}+w_T \times h_P-1), \bar{P'})$$
The claim follows by defining $\lambda(P) = \bar{P'}$ and $\lambda(T) = \bar{T}$.  
\end{proof}

\begin{corollary}\label{cor:sigman2d}
Given two non-empty two-dimensional strings $P$ and $T$ of widths $w_P, w_T$ and heights~$h_P, h_T$, we can solve the text-to-pattern Hamming distances problem in time $\O(\sigma \cdot N \log N)$, where $\sigma$ is the number of different characters present both in $P$ and $T$ and $N = (w_P + w_T)(h_P + h_T)$.
\end{corollary}
\begin{proof}
We replace all characters of $T$ non-present in $P$ with a new character in $\O(w_T\times h_T)$ time. The Hamming distances do not change. The claim then follows immediately from \cref{fact:sigman1d} and \cref{claim:padding}. 
\end{proof}

\begin{corollary}\label{cor:approx2d}
Let $\varepsilon > 0$ be a constant. Given two non-empty two-dimensional strings $P$ and $T$ of widths $w_P, w_T$ and heights $h_P, h_T$, there is a $(1+\varepsilon)$-approximation algorithm that solves the text-to-pattern Hamming distances problem in time $\O((N/\varepsilon^2)  \log^3 N)$, where $N = (w_P + w_T)(h_P + h_T)$.
\end{corollary}
\begin{proof}
Follows from \cref{cor:approx1d} and \cref{claim:padding}. 
\end{proof}


\section{Technical overview}
In this section, we present the technical ideas behind our main result:

\def\mainThmTitle{Formal version}
\def\mainThmContent{%
Given two-dimensional strings $P$ with $\d{P} = [m] \times [m]$ and $T$ with $\d{T} = [n] \times [n]$, where $m,n\in \Z^+$ and $m \le n$. There is an algorithm that solves the $k$-mismatch problem for $P,T$ in $\tO((m^2 + mk^{5/4})n^2 / m^2)$ time.%
}

\restateThmMain*

Assume w.l.o.g.~that $2|n$ and $n \le \frac{3}{2}m$. If this is not the case, one can cover $T$ with $\O(n^2/m^2)$ strings $T_i$ such that $\d{T_i} = [n'] \times [n'] + q$ for $n' \le \frac{3}{2}m$, $2|n'$, and $q \in \Z^2$ and solve the problem for $P$ and each $T_i$ independently, which will result in $\O(n^2/m^2)$ extra multiplicative factor in the time complexity. Hence, below we assume that the two conditions are satisfied and aim to design an $\tO(m^2 + mk^{5/4})$-time algorithm. 

We start by applying \cref{cor:approx2d} to $P$, $T$, and $\varepsilon = 1$ and obtain in $\O(n^2 \log^3 n)$ time a set $Q \subseteq \Z^2$ such that both of the following hold:
\begin{enumerate}
\item For all $q \in \Z^2$ such that $\Ham(T, P+q) \le k$, we have $q \in Q$;
\item For all $q \in Q$, we have $\Ham(T, P+q) \le 2 k$.
\end{enumerate}

If $\absolute{Q} \le 8m + m^2/k$, we apply \cref{kangaroos} to compute the true value of the Hamming distance for all $q \in Q$, which takes $\O(mk+m^2)$ total time, completing the proof of \cref{th:main}. Below we assume $\absolute{Q} > 8m + m^2/k$. 

The idea of the algorithm is to partition the pattern $P$ and the text $T$ into strings of regular form each containing a single character. Intuitively, computing the Hamming distance between two such strings is easy: if the characters are equal, the distance is zero, and otherwise, it is the size of their intersection.

\subsection{Two-dimensional periodicity.}
We start by defining two-dimensional periods.


\begin{restatable*}[Two-dimensional approximate period]{definition}{restateDefTwoDApproxPeriod}
We say that a string $S$ has an \emph{$\ell$-period} $\delta \in \Z^2$ when $\Ham(S + \delta, S) \le \ell$.
\end{restatable*}

\begin{restatable*}{claim}{restateClaimPeriodicity} \label{periodicity_lemma}
For all $u, v \in Q$, we have that $u - v \in \Z^2$ is a $\O(k)$-period of $P$.
\end{restatable*}	
\begin{proof}
$\Ham(P + u - v, P) = \Ham(P + u, P + v) \le \Ham(P + u, T) + \Ham(T,P + v) = \O(k)$.
\end{proof}

As $Q$ is big, it generates two $\O(k)$-periods with a big angle between them and such that they and their negations belong to distinct quadrants of $\Z^2$: 

\begin{restatable*}{@theorem}{restateThmGetPeriods}\label{get_periods}
Given $\ell \in \Z^+$ and $U \subseteq [\ell + 1]^2$ of size $> 16\ell$. There is an $\tO(\absolute{U})$-time algorithm that finds $u, v, u', v' \in U$, such that the following conditions hold for $\psi = u - v$ and $\phi = u' - v'$:
	\begin{itemize}
		\item $\psi \in (0, +\infty) \times [0, +\infty)$, $\phi \in [0, +\infty) \times (-\infty, 0)$;
		\item $0 < \absolute{\psi}\absolute{\phi} = \O(\ell^2 / \absolute{U})$,
		\item $\sin \alpha \ge \frac{1}{2}$, where $\alpha$ is the angle between $\psi$ and $\phi$.
	\end{itemize}
\end{restatable*}

Let $\ell = n - m \le m / 2$. We apply \Cref{get_periods} on $\ell$ and the set $Q$ (the conditions of the claim are satisfied as $\absolute{Q} > 8m + m^2/k \ge 16\ell$) and compute $\phi, \psi \in \Z^2$ in $\tO(\absolute{Q})$ time. By \Cref{periodicity_lemma}, $\phi$ and $\psi$ are $\O(k)$-periods of $P$, and $0 \le \phi \times \psi \le \absolute{\phi}\absolute{\psi} = \O(\ell^2 / \absolute{Q}) =  \O(\min\set{m, k})$. We fix $\phi$ and $\psi$ for the rest of the paper.

\subsection{Partitioning of the pattern and the text.}
Next, we define tiles, which will be used to partition $P$ and $T$ into one-character strings of regular form.

\begin{restatable*}[Lattice congruency]{definition}{restateDefLatticeCongruency}\label{lattice_congruency}
We define \emph{a lattice} $\L = \set{s\phi + t\psi : s, t \in \Z}$. We say that $u, v \in \Z^2$ are \emph{congruent}, denoted $u \equiv v$, if $u - v \in \L$. 
\end{restatable*}

\begin{figure}
	\begin{center}
		\includegraphics[width=0.8\textwidth]{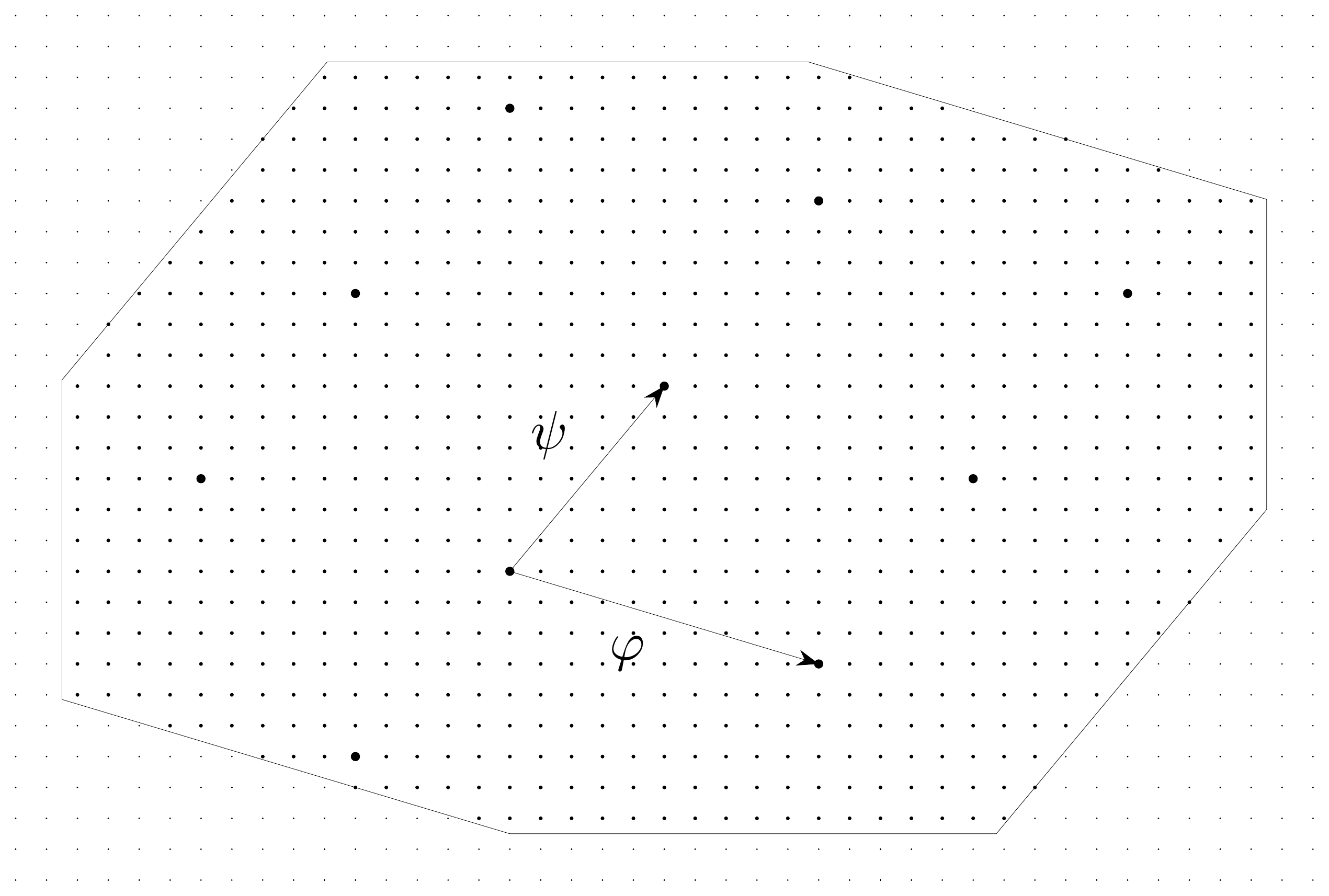}
	\end{center}
	\caption{All the points in the polygon form a truncated tile and the thicker points form a truncated subtile.}
	\label{figure:tile}
\end{figure}

\begin{restatable*}[Tile]{definition}{restateDefTile}\label{tile_definition}
We call $U \subseteq \Z^2$ a \emph{tile} if there exist $\phi_0, \phi_1, \psi_0, \psi_1 \in \Z$, such that
	\[ U = \set{u : u \in \Z^2, \h{u} \in [\phi_0, \phi_1], \s{u} \in [\psi_0, \psi_1]}. \]
We call the values $\phi_0, \phi_1, \psi_0, \psi_1$ defining $U$ its \emph{signature}. 
\end{restatable*}

Geometrically, a tile $U$ can be viewed as a set of integer points inside a parallelogram with sides parallel to $\phi$ and $\psi$. Truncated tiles are a generalisation of simple tiles. 

\begin{restatable*}[Truncated tile]{definition}{restateDefTruncatedTile}
We call $U \subseteq \Z^2$ a \emph{truncated tile} if there exist values $x_0, x_1, y_0, y_1 \in \Z \cup \{-\infty, +\infty\}$ and $\phi_0, \phi_1, \psi_0, \psi_1 \in \Z$ such that 
	\begin{enumerate}
		\item $U = [x_0, x_1] \times [y_0, y_1] \cap \set{u : u \in \Z^2, \h{u} \in [\phi_0, \phi_1], \s{u} \in [\psi_0, \psi_1]}$. 
		\item $x_1 - x_0 + 1 \ge \absolute{\x{\phi}} + \absolute{\x{\psi}}$ and $y_1 - y_0 + 1 \ge \absolute{\y{\phi}} + \absolute{\y{\psi}}$. 
	\end{enumerate}
The values $x_0, x_1, y_0, y_1, \phi_0, \phi_1, \psi_0, \psi_1$ defining $U$ are called its \emph{signature}. 
\end{restatable*}

Geometrically, the first property of a truncated tile $U$ means that it is a set of integer points in the intersection of an axis-parallel rectangle (possibly infinite) and a parallelogram with sides parallel to $\phi$ and $\psi$. The second property implies that the axis-parallel rectangle defined by $x_0, x_1, y_0, y_1$ contains a parallelogram spanned by $\phi$ and $\psi$. It will become clear later why this property is important. 
Note that every tile is a truncated tile.

\begin{restatable*}[Subtile]{definition}{restateDefSubtile}\label{subtile_definition}
We call a set $V \subseteq \Z^2$ a \emph{(truncated) subtile} if there exists a (truncated) tile $U$ and $\gamma \in \Z^2$ such that $V = \set{u : u \in U, u \equiv \gamma}$. A signature of $V$ consists of the signature of $U$ and $\gamma$. Abusing notation, we write $V \equiv v$ for $v \in \Z^2$ if $v \equiv \gamma$ and for two (truncated) subtiles $U,V$ we write $U \equiv V$ if there exists $v \in \Z^2$ such that $U \equiv v$ and $V \equiv v$.
\end{restatable*}

See \Cref{figure:tile} for an illustration. Note that every subtile is a truncated subtile.

\begin{restatable*}[Tile string]{definition}{restateDefTileString}\label{tile_string_definition}
We call a string $S$ a \emph{(truncated) (sub-)tile string} if $\d{S}$ is a (truncated) (sub-)tile.
\end{restatable*}

\begin{restatable*}[Monochromatic string]{definition}{restateDefMonochromaticString}
A string $S$ is called \emph{monochromatic} if there is a character $\getchar{S} \in \Sigma$ such that for every $u \in \d{S}$ there is $S(u) = \getchar{S}$.  
\end{restatable*}

\begin{restatable*}{@theorem}{restateThmTileDecomposition}\label{tile_decomposition}
Assume to be given a (truncated) tile string $R$ such that $\phi, \psi$ are its $\O(k)$-periods. After $\O(m^2)$-time preprocessing of $\phi, \psi$, the string $R$ can be partitioned in time $\tO(\absolute{\d{R}} + k)$ into $\O(k)$ monochromatic (truncated) subtile strings.
\end{restatable*}

Since $\absolute{\x{\phi}}, \absolute{\y{\phi}}, \absolute{\x{\psi}}, \absolute{\y{\psi}} \le n - m \le m / 2$, the string $P$ is a truncated tile string, and we can apply \Cref{tile_decomposition} to partition it in time $\tO(m^2 + k)$ into a set of monochromatic truncated subtile strings $\V$. We then group the strings in $\V$ based on the single character they contain. Specifically, for every character $a \in \Sigma$ present in $P$, we construct a set $\V_a = \set{V : V \in \V, \getchar{V} = a}$.

Because $T$ is not necessarily periodic, we need to use a similar but more sophisticated partitioning.

\begin{restatable*}[Active text]{definition}{restateDefActiveText}
We define the \emph{active text} $\Ta$ as the restriction of $T$ to $\bigcup_{q \in Q} \d{P + q}$. 
\end{restatable*}

\begin{restatable*}{observation}{restateObsActiveText}\label{obs:active_text}
$\Ham(P + q, T) = \Ham(P + q, \Ta)$ for every $q \in Q$.
\end{restatable*}

\begin{restatable*}[Peripherality]{definition}{restateDefPeripherality}
For every point $u \in \Z^2$, we define its \emph{border distance} as $\min\set{\absolute{u - v} : v \in \Z^2 \setminus \d{\Ta}}$. A set $U \subseteq \Z^2$ is $d$-\emph{peripheral} for some $d \ge 0$, if the border distance of every $u \in U$ is at most $d$. We say that a string $S$ is $d$-peripheral when $\d{S}$ is $d$-peripheral.
\end{restatable*}

\begin{restatable*}{@theorem}{restateThmTextDecomposition}\label{text_decomposition}
Given $\ell \in \Z^+$, one can partition $\Ta$ in time $\tO(m^2 + \ell^2 + \ell k)$ into $\O(\ell k)$ monochromatic subtile strings and an $\O(m / \ell)$-peripheral string.
\end{restatable*}

\subsection{Algorithm.}
In \cref{sec:algorithm}, we show that the Hamming distances between $P$ and a string that can be partitioned into a set of monochromatic subtile strings can be computed efficiently: 

\begin{restatable*}{@theorem}{restateThmSparseAlgo}
\label{th:sparse_algo}
Let $\S$ be a set of monochromatic subtile strings $\S$ with a property that the domains $\d{S}$ for $S \in \S$ are pairwise disjoint subsets of $\d{T}$. There is an algorithm that computes
$\sum_{S \in \S} \Ham(P + q, S)$ for every $q \in Q$ in total time $\tO(m^2 + \sum_{S \in \S} \absolute{\V_{\getchar{S}}})$.
\end{restatable*}

As an immediate corollary, we obtain an $\tO(m^2 + mk^2)$-time solution for the $k$-mismatch problem. First, we apply \Cref{text_decomposition} for a large enough value $\ell = \Theta(m)$ to partition the active text $\Ta$ into a set $\S$ of $\O(mk)$ monochromatic subtile strings and a $0$-peripheral string, empty by definition, in time $\tO(m^2 + mk)$. For every $q \in Q$ we then have $\Ham(P + q, \Ta) = \sum_{S \in \S} \Ham(P + q, S)$, and these values can be computed in $\tO(m^2 + mk^{2})$ total time by \cref{th:sparse_algo}, since $\sum_{S \in \S} \absolute{\V_{\getchar{S}}} \le \absolute{\S} \absolute{\V} = \O(mk \cdot k)$. 

To improve the complexity further, we partition the active text using the algorithm from \Cref{text_decomposition} with $\ell = mk^{-3/4}$.
We obtain a set~$\S$ of $\O(mk^{1/4})$ monochromatic subtile strings $\S$, and an $\O(k^{3 / 4})$-peripheral string $F$. For every $q \in Q$ we then have

\[ \Ham(P + q, \Ta) = \Ham(P + q, F) + \sum_{S \in \S} \Ham(P + q, S).\]

By \cref{th:sparse_algo}, we can compute $\sum_{S \in \S} \Ham(P + q, S)$ for every $q \in Q$ in time $\tO(m^2 + mk^{5/4})$, since similarly to above we have $\sum_{S \in \S} \absolute{\V_{\getchar{S}}} \le \absolute{\S} \absolute{\V} = \O(mk^{5/4})$. To compute the values $\Ham(P + q, F)$, we apply the following result:

\begin{restatable*}{@theorem}{restateThmDenseAlgo}
\label{th:dense_algo}
Given a $d$-peripheral string $F$, there is an algorithm that computes $\Ham(P + q, F)$ for every $q \in Q$ in total time $\tO(m^2 + mdk^{1/2})$.
\end{restatable*}

By substituting $d = \O(k^{3/4})$, we obtain the desired time complexity of $\tO(m^2 + mk^{5/4})$ and hence complete the proof of \cref{th:main}.
In \cref{sec:two_periods,sec:partition,sec:algorithm}, we explain the proof in full detail, including the proofs of the main technical results stated in \cref{get_periods,tile_decomposition,text_decomposition,th:sparse_algo,th:dense_algo}.

\section{Two-dimensional periodicity}
\label{sec:two_periods}

Now we provide a full proof of the main result, which is restated below. In this section, we show that we can easily identify a superset of the $k$-mismatch occurrences of the pattern, where every false positive is still a $2k$-mismatch occurrence. If this superset is small, then we can easily filter out the false positives and thus solve the problem. Otherwise, we show that the pattern is close to being highly periodic, and we can find two approximate periods that have additional useful properties.

\restateThmMain*

Assume w.l.o.g.~that $2|n$ and $n \le \frac{3}{2}m$. If this is not the case, one can cover $T$ with $\O(n^2/m^2)$ strings $T_i$ such that $\d{T_i} = [n'] \times [n'] + q$ for $n' \le \frac{3}{2}m$, $2|n'$, and $q \in \Z^2$ and solve the problem for $P$ and each $T_i$ independently, which will result in $\O(n^2/m^2)$ extra multiplicative factor in the time complexity. Hence, below we assume that the two conditions are satisfied and aim to design an $\tO(m^2 + mk^{5/4})$-time algorithm. 

We start by applying \cref{cor:approx2d} to $P$, $T$, and $\varepsilon = 1$ and obtain in $\O(n^2 \log^3 n)$ time a set $Q \subseteq \Z^2$ such that both of the following hold:
\begin{enumerate}
\item For all $q \in \Z^2$ such that $\Ham(T, P+q) \le k$, we have $q \in Q$;
\item For all $q \in Q$, we have $\Ham(T, P+q) \le 2 k$.
\end{enumerate}

If $\absolute{Q} \le 8m + m^2/k$, we apply \cref{kangaroos} to compute the true value of the Hamming distance for all $q \in Q$, which takes $\O(mk+m^2)$ total time, completing the proof of \cref{th:main}. Below we assume $\absolute{Q} > 8m + m^2/k$. 

The idea of the algorithm is to partition the pattern $P$ and the text $T$ into strings of regular form each containing a single character. Intuitively, computing the Hamming distance between two such strings is easy: if the characters are equal, the distance is zero, and otherwise, it is the size of their intersection.
In order to achieve the desired time complexity, it is crucial that the number of regular strings that we partition $P$ and $S$ into is sufficiently small.
We achieve this by showing that the pattern is approximately periodic according to the following definition.

\restateDefTwoDApproxPeriod
\restateClaimPeriodicity
\begin{proof}
$\Ham(P + u - v, P) = \Ham(P + u, P + v) \le \Ham(P + u, T) + \Ham(T,P + v) = \O(k)$.
\end{proof}

As $Q$ is big, it generates many $\O(k)$-periods of $P$. Below, we show that there must be two $\O(k)$-periods belonging to neighbouring quadrants of $\Z^2$, and such that the angle between them is between $30$ and $150$ degrees.

\restateThmGetPeriods
\begin{proof}
We start by finding the closest pair of points in $U$.
Specifically, we select any pair of different points $s, t \in U$ that minimizes $\absolute{t - s}$. 
Such a pair can be obtained in $\O(\absolute{U} \log \absolute{U})$ time~\cite{Clarkson1983}.
We construct $w = t - s$.
Denote the quadrants of the plane as follows:%
\begin{align*}
&\Q_1 = (0,+\infty) \times [0,+\infty), \quad & \Q_2 = (-\infty,0] \times (0,+\infty),\\
&\Q_3 = (-\infty,0) \times (-\infty,0], \quad & \Q_4 = [0,+\infty) \times (-\infty,0).
\end{align*}
We will compute a vector $w' = t' - s'$ with $s', t' \in U$ and the following properties. If $w \in \Q_i$, then $w' \in \Q_{(i \bmod 4) + 1}$, i.e., $w'$ is in the counterclockwise neighbouring quadrant of $w$. Further, it holds $0 < \absolute{w}\absolute{w'} = \O(\ell^2 / \absolute{U})$, and $\sin \alpha \ge \frac{1}{2}$, where $\alpha$ is the angle between $w$ and $w'$. Depending on $i$, we then report $\psi$ and $\phi$ as follows. If $i = 1$, then $\psi = w$ and $\phi = -w'$. If $i = 2$, then $\psi = -w'$ and $\phi = -w$. If $i = 3$, then $\psi = -w$ and $\phi = w'$. If $i = 4$, then $\psi = w'$ and $\phi = w$. Clearly, this satisfies the claimed properties, and we can easily output the corresponding nodes $u,v,u',v'$.

Assume w.l.o.g.\ that $w \in \Q_1$ (the other cases are entirely symmetric). 
We extend $\Q_1$ to $\Q_1^+$ by adding the neighbouring $30$ degrees of $\Q_2$ and $\Q_4$, and we restrict $\Q_2$ to $\Q_2^-$ by limiting it to points strictly contained in its central 30 degrees, formally

\begin{center}
\def\fontscaler{.8}
\def\imagescaler{.55}
\let\oldrotatebox\rotatebox
\renewcommand{\rotatebox}[2]{\oldrotatebox{#1}{\scalebox{\fontscaler}{#2}}}
\begin{tikzpicture}[x=\imagescaler em, y=\imagescaler em] 

\pgfmathparse{10/sqrt(3)}
\edef\line{\pgfmathresult}

\draw[thick] (-10, \line) to (10, -\line);
\draw[thick] (-\line, 10) to (\line, -10);

\fill[gray!30!white] (-10,-10) rectangle (10,10);
\fill[gray!10!white] (-10, \line) to (-10, 10) to (-\line, 10) to (0,0) to (-10, \line);
\fill[gray!10!white] (10, -\line) to (10, -10) to (\line, -10) to (0,0) to (10, -\line);

\path (10,10) node[below left] {$\Q_1^+$};
\path (-10,10) node[below right] {$\Q_2^-$};
\path (-10,-10) node[above right] {$\Q_3^+$};
\path (10,-10) node[above left] {$\Q_4^-$};

\draw[thick, blue] (-10, \line) to (0,0);
\draw[thick, blue] (-\line, 10) to (0,0);
\draw[thick, blue] (0,0) to (10, -\line);
\draw[thick, blue] (0,0) to (\line, -10);

\node[circle, fill=white, inner sep=1pt] at (0,0) {};
\draw[densely dotted, -Latex] (0,-10) to node[pos=1, below right] {$y$} (0,+10);
\draw[densely dotted, -Latex] (-10,0) to node[pos=1, above left] {$x$} (+10,0);

\path (10, -\line) to node[fill=gray!30!white, pos=1, above right=0 and .2em, rotate=-30, anchor=south west] {\scalebox{\fontscaler}{\enskip $y = -x/\sqrt{3}$}} (0,0);

\path (\line, -10) to node[fill=gray!30!white, pos=1, below right=.2em and 0, rotate=-60, anchor=north west] {\scalebox{\fontscaler}{\enskip $y = -x \cdot \sqrt{3}$}} (0,0);

\draw[gray] (-7,0) to[out=90, in=180] %
node[pos=.0] (t0) {}
node[pos=.03333] (t1) {}
node[pos=.33333] (t2) {} 
node[pos=.36667] (t3) {} 
node[pos=.66667] (t4) {}
node[pos=.7] (t5) {}
node[pos=.1666666] (l1) {}
node[pos=.5] (l2) {}
node[pos=.8333333] (l3) {}
(0,7);

\draw[gray, -Latex] (t1.center) to (t0.center);
\draw[gray, -Latex] (t3.center) to (t2.center);
\draw[gray, -Latex] (t5.center) to (t4.center);

\node[gray] at (l1) {\rotatebox{-15}{\rlap{\ $30^{\circ}$}}};
\node[gray] at (l2) {\rotatebox{-45}{\rlap{\ $30^{\circ}$}}};
\node[gray] at (l3) {\rotatebox{-75}{\rlap{\ $30^{\circ}$}}};

\draw (-10,0) node[above left=.25em and 1em] {%
$\Q_1^+ = \{ (x,y) \in \mathbb Z^2{\ \setminus \{(0,0)\}} \mid y \geq -x / \sqrt{3}\text{ and }y \geq -x \cdot \sqrt{3} \},$%
};

\draw (-10,0) node[below left=.25em and 1em] {%
$\Q_2^- = \{ (x,y) \in \mathbb Z^2\phantom{\ \setminus \{(0,0)\}} \mid y > -x / \sqrt{3}\text{ and }y < -x \cdot \sqrt{3}\}.$%
};

\node[inner sep=0pt] (w) at (2,5) {};
\draw[-Latex] (0,0) to (w) node[above right] {$w$};

\end{tikzpicture}
\end{center}

%

%
Analogously, we obtain the extended $\Q_3$ as $\Q_3^+ = \{ z \in \mathbb Z^2 \mid -z \in \Q_1^+ \}$, and the restricted $\Q_4$ as $\Q_4^- = \{ z \in \mathbb Z^2 \mid -z \in \Q_2^- \}$. The angle between a vector in $\Q_1$ (not extended) and a vector in $\Q_2^-$ (restricted) is clearly at least $30$ degrees and at most $150$ degrees. We will find $w'$ in $\Q_2^-$, and thus the angle $\alpha$ between $w$ and $w'$ will satisfy $\sin(\alpha) \geq \frac12$ by design.

We define two relations $<_1$ and $<_2$ over $U$. For $u, v \in U$, it holds $u <_1 v$ if and only if $(v - u) \in \Q_1^+$, and $u <_2 v$ if and only if $(v - u) \in \Q_2^-$. In a moment, we will show that both $<_1$ and $<_2$ are strict partial orders. 
Note that two distinct points $u,v \in U$ are comparable with $<_1$ if and only if $(v - u) \in \Q_1^+ \cup \Q_3^+$, and they are comparable with $<_2$ if and only if $(v - u) \in \Q_2^- \cup \Q_4^-$. 
By the definition of the sets, it is clear that $\Q_1^+ \cup \Q_3^+$ and $\Q_2^- \cup \Q_4^-$ form a bipartition of $\mathbb Z^2 \setminus \{(0,0)\}$.
Thus, it is clear that every antichain with respect to $<_1$ is a chain with respect to $<_2$ and vice versa.

Now we show that $<_1$ is a strict partial order. The proof for $<_2$ works analogously. First, note that $(0,0) \notin \Q_1^+$, and thus it cannot be that $u <_1 u$, i.e., the relation is \emph{irreflexive}. It is easy to see that the intersection of $\Q_1^+$ and $\Q_3^+$ is empty. This means that, if $u <_1 v$, then it cannot be that $v <_1 u$, i.e., the relation is \emph{asymmetric}. Finally, given $u <_1 v$ and $v <_1 v'$, it holds $(v.y - u.y) \geq -(v.x-u.x)/\sqrt{3}$ and $(v'.y - v.y) \geq -(v'.x-v.x)/\sqrt{3}$ by the definition of $\Q_1^+$. By summing the inequalities, we get $(v'.y - u.y) \geq -(v'.x-u.x)/\sqrt{3}$. The same reasoning applies to the second inequality of the definition. Therefore, it holds $(v' - u) \in \Q_1^+$ and hence $u <_1 v'$, i.e., the relation is \emph{transitive}.

We have shown that $<_1$ and $<_2$ are strict partial orders, with a one-to-one correspondence between antichains of $<_1$ and chains of $<_2$.
Let $C$ be the longest chain of $<_1$, and let $A$ be the longest antichain of $<_1$, or equivalently the longest chain of $<_2$. First, we show upper bounds for $\absolute{C}$ and $\absolute{A}$ using Dilworth's theorem:

\begin{fact}[Dilworth's theorem]\label{dilworth}
	$\absolute{U} \le \absolute{C} \absolute{A}$.
\end{fact}

\begin{claim}\label{C_ineq}
	$\absolute{C} \cdot \absolute{w} \leq 16\ell$. 
	\begin{proof}
		Let $f = \absolute{C} - 1$ and
		let $c_0, \dots, c_{f}$ denote the points in $C$, ordered so that $c_i <_1 c_{i + 1}$ for every $i \in [f]$.
		Consider any $u = c_i$ and $v = c_{i + 1}$ with $i \in [f]$. We will show that $v.x + v.y \geq u.x + u.y + \frac14\absolute{w}$. 		
		Let $w'' = (v - u)$, which is in $\Q_1^+$ due to $u <_1 v$. Assume $v.x + v.y < u.x + u.y + \frac14\absolute{w}$, or equivalently $(w''.x + w''.y) < \frac14\absolute{w}$. From now on, let $x = w''.x$ and $y = w''.y$. It then holds $$\absolute{w''} \geq \absolute{w} > 4(x + y).$$
		
		This is already a contradiction if both $x$ and $y$ are non-negative, since then it should hold $\sqrt{x^2 + y^2} \leq \absolute{x} + \absolute{y}$. If both $x$ and $y$ are negative, then they contradict $y \geq -x/\sqrt{3} \geq 0$ in the definition of $\Q_1^-$. Hence exactly one of them is negative. If $x \geq 0$ and $y < 0$, then $y \geq -x / \sqrt{3}$ in the definition of $\Q_1^-$ implies $\absolute{x} \geq \absolute{y}\cdot \sqrt{3}$ and thus $\frac34\absolute{x} > \frac54\absolute{y}$. Similarly, if $x < 0$ and $y \geq 0$, then $y \geq -x \cdot \sqrt{3}$ implies $\frac34\absolute{y} > \frac54\absolute{x}$. Hence we have
	\begin{alignat*}{1}	
		\absolute{w''} > 4(x + y) = 
		4(\absolute{x} - \absolute{y}) > 
		4(\tfrac14\absolute{x} + \tfrac14\absolute{y}) = 
		\absolute{x} + \absolute{y}&\text{\quad if $x \geq 0$ and $y < 0$, and}\\
		\absolute{w''} > 4(x + y) = 
		4(\absolute{y} - \absolute{x}) > 
		4(\tfrac14\absolute{y} + \tfrac14\absolute{x}) = 
		\absolute{x} + \absolute{y}&\text{\quad if $x < 0$ and $y \geq 0$,}
	\end{alignat*}
	which is a contradiction.
	We have shown that $c_{i + 1}.x + c_{i + 1}.y \geq c_i.x + c_i.y + \frac14\absolute{w}$ for all $i \in [f]$. By chaining the inequality for all $i$, we get $$2\ell \geq c_f.x + c_f.y \geq c_0.x + c_0.y + \sum\nolimits_{i = 0}^{f - 1} \tfrac14\absolute{w} \geq \tfrac14f\cdot{\absolute{w}},$$ and thus $\absolute{w} \cdot f \leq 8\ell$. Note that $w = (t - s) \in \Q_1^+$, and thus $t$ and $s$ belong to a chain, which implies $\absolute{C} > 1$. Thus $\absolute{w} \cdot \absolute{C} \leq 2\absolute{w} \cdot f \leq 16\ell$, which concludes the proof.
	\end{proof}%
\end{claim}

By the assumption $\absolute{U} \ge 16\ell$ and \Cref{dilworth}, it holds
$16 \ell \le \absolute{U} \le \absolute{C} \absolute{A} < 16\ell\absolute{A}$,
thus $\absolute{A} \ge 2$.
We select any pair of different vectors $s', t' \in A$ that minimizes $\absolute{t' - s'}$. Since $s'$ and $t'$ are in a chain of $<_2$, it holds either $s' <_2 t'$ or $t' <_2 s'$. W.l.o.g., we assume $s' <_2 t'$ (otherwise swap $s'$ and~$t'$), and define $w' = t' - s'$, which is in $\Q_2^-$ as required.  We show how to compute $A$ and $w'$ later.

\begin{claim}\label{A_ineq}
	$\absolute{A}\cdot \absolute{w'} \leq 16\ell$. 
	\begin{proof}
	
		Let $f = \absolute{C} - 1$ and
		let $c_0, \dots, c_{f}$ denote the points in $A$, ordered so that $c_i <_2 c_{i + 1}$ for every $i \in [f]$ (recall that $A$ is a chain with respect to $<_2$).
		Consider any $u = c_i$ and $v = c_{i + 1}$ with $i \in [f]$. We will show that $v.x + v.y \geq u.x + u.y + \frac14\absolute{w'}$. 		
		Let $w'' = (v - u)$, which is in $\Q_2^-$ due to $u <_2 v$. From now on, let $x = w''.x$ and $y = w''.y$. 
		Note that $-x/\sqrt{3} < y < -x \sqrt{3}$ in the definition of $\Q_2^-$ implies that $x$ is negative and $y$ is positive.
		Therefore, $y > -x/\sqrt{3}$ implies $\sqrt{3}\absolute{y} < \absolute{x}$, and thus $\frac 5 4\absolute{y} < \frac 3 4\absolute{x}$.
		Assume $v.x + v.y < u.x + u.y + \frac14\absolute{w'}$, or equivalently $(\absolute{x} - \absolute{y}) = (x + y) < \frac14\absolute{w'}$. Then $$\absolute{w''} \geq \absolute{w'} > 4(x + y) = 4(\absolute{x} - \absolute{y}) > 4(\tfrac14\absolute{x} + \tfrac14\absolute{y}) = 
		\absolute{x} + \absolute{y},$$
		which contradicts $\sqrt{x^2 + y^2} \leq \absolute{x} + \absolute{y}$. We have shown that $c_{i + 1}.x + c_{i + 1}.y \geq c_i.x + c_i.y + \frac14\absolute{w}$ for all $i \in [f]$. Hence we can finish the proof exactly like in \cref{C_ineq} and obtain $\absolute{A}\cdot \absolute{w'} \leq 16\ell$.
	\end{proof}
\end{claim}

By multiplying the inequalities from \cref{A_ineq,C_ineq}, we obtain $\absolute{A}\absolute{C}\absolute{w}\absolute{w'} \leq 256\ell^2$. Finally, applying \cref{dilworth} and dividing both sides by $\absolute{U}$, we obtain the final statement $\absolute{w}\absolute{w'} \leq 256\ell^2/\absolute{U}$.

We have shown that $w$ and $w'$ satisfy the claimed properties. 
It remains to show how to compute~$w'$. 
If the longest antichain $A$ is given, then it is easy to sort the chain according to $<_2$ and select the desired $w'$ in $\tO(\absolute{A})$ time. To compute the antichain, 
we define two orders $\prec_a$ and $\prec_b$, where 
$u \prec_a v$ if and only if $\sqrt{3} \cdot u.y + u.x < \sqrt{3} \cdot v.y + v.x$, and $u \prec_b v$ if and only if $\sqrt{3} \cdot v.x + v.y < \sqrt{3} \cdot u.x + u.y$.
By the definition of $Q_2^-$, it holds $u <_2 v$ if and only if both $u \prec_a v$ and $u \prec_b v$.
We sort $U$ in non-decreasing order with respect to $\prec_a$. If there is a tie between $u$ and $v$, we break the tie such that $u$ precedes $v$ if and only if $v \prec_b u$. (There cannot be distinct $u$ and $v$ that are tied with respect to both orders.) 
Afterwards, a subsequence of the sorted $U$ is a chain of $<_2$ if and only if the subsequence (read from left to right) is increasing with respect to $\prec_b$.
We compute the longest increasing subsequence with respect to $\prec_b$ in $\tO(\absolute{U})$ time (see~\cite{CROCHEMORE20101054} and references therein) and hence obtain the longest chain of $<_2$, or equivalently the longest antichain of $<_1$.
\end{proof}

Let $\ell = n - m \le m / 2$. We apply \Cref{get_periods} on $\ell$ and the set $Q$ (the conditions of the claim are satisfied as $\absolute{Q} > 8m + m^2/k \ge 16\ell$) and compute $\phi, \psi \in \Z^2$ in $\tO(\absolute{Q})$ time. By \Cref{periodicity_lemma}, $\phi$ and $\psi$ are $\O(k)$-periods of $P$, and $0 \le \phi \times \psi \le \absolute{\phi}\absolute{\psi} = \O(\ell^2 / \absolute{Q}) =  \O(\min\set{m, k})$. We fix $\phi$ and $\psi$ for the rest of the paper.

\section{Partitioning of the pattern and the text}
\label{sec:partition}

As mentioned before, we will partition $P$ and $T$ into one-character strings of regular form. A one-character string is formally defined as follows.

\restateDefMonochromaticString

The regular form of the strings is obtained by cutting the domains of $P$ and $T$ into so-called \emph{(truncated) tiles}.

\restateDefTile

Geometrically, a tile $U$ can be viewed as a set of integer points inside a parallelogram with sides parallel to $\phi$ and $\psi$. Truncated tiles are a generalisation of simple tiles. 

\restateDefTruncatedTile

Geometrically, the first property of a truncated tile $U$ means that it is a set of integer points in the intersection of an axis-parallel rectangle (possibly infinite) and a parallelogram with sides parallel to $\phi$ and $\psi$. The second property implies that the axis-parallel rectangle defined by $x_0, x_1, y_0, y_1$ contains a parallelogram spanned by $\phi$ and $\psi$. It will become clear later why this property is important. 

\subsection*{Partitioning tiles into subtiles.}

Using the lattice defined by basis vectors $\phi$ and $\psi$, we can further partition a (truncated) tile into $\O(\min\{m, k\})$ (truncated) \emph{subtiles}. This is done by dividing the points into congruence classes depending on their relative position with respect to the lattice.

\restateDefLatticeCongruency

\begin{lemma} \label{lattice_base}
	There exists an algorithm that constructs in $\O(m^2)$ time a set $\Gamma \subseteq \Z^2$ such that $\absolute{\Gamma} = \O(\min\set{m, k})$ and $\forall u \in \Z^2 \; \exists \gamma \in \Gamma : u \equiv \gamma$.
\end{lemma} 
	\begin{proof}
		Let $p = \set{s\phi + t\psi : s \in [0, 1), t \in [0, 1)}$ and $\Gamma = p \cap \Z^2$.
		By Pick's Theorem, a simple polygon with integer vertices contains $\O(A)$ integer points in the interior or on the boundary, where $A$ denotes its area.
		Observe that the elements of $\Gamma$ are contained in a parallelogram with vertices $(0, 0), \phi, \phi + \psi, \psi$.
		Since its area is $\phi \times \psi = \O(\min\set{m, k})$, we get $\absolute{\Gamma} = \O(\min\set{m, k})$.
		
		Now consider any $u \in \Z^2$.
		As $\phi, \psi$ are not collinear, there exist unique values $s, t \in [0, 1)$ and $s', t' \in \Z$, such that
		$u = (s + s') \phi + (t + t') \psi$.
		By definition, 
		$u \equiv s\phi + t\psi$ and $s\phi + t\psi \in \Gamma$.
		
		We now give a construction algorithm for $\Gamma$.
		Consider $\Gamma' = \Z^2 \cap [0,\x{\phi}+\x{\psi}] \times [\y{\phi},\y{\psi}]$.
		We have $\absolute{\Gamma'} = \O((\absolute{\x{\phi}}+\absolute{\x{\psi}}) \cdot (\absolute{\y{\phi}}+\absolute{\y{\psi}})) = \O(m^2)$ as $\absolute{\x{\phi}}, \absolute{\y{\phi}}, \absolute{\x{\psi}}, \absolute{\y{\psi}} \le n - m \le m / 2$.
		For $u \in \Gamma'$ we have $u \in \Gamma$ if and only if $\s{\phi} \le \s{u} \le 0$ and $0 \le \h{u} \le \h{\psi}$, which can be checked in constant time. 
	\end{proof}

\restateDefSubtile

See \Cref{figure:tile} for an illustration.

\restateDefTileString

\subsection{Pattern partitioning.}
\label{sec:partition:pattern}

In this section, we show how to partition the pattern $P$ into $\O(k)$ monochromatic truncated subtile strings. The procedure works not only for $P$, but for any truncated tile string with $\O(k)$-mismatch periods $\psi$ and $\phi$. Later, we use the same procedure to also partition a part of the text.

\begin{definition}[Lattice graph]
	For a set $U \subseteq \Z^2$ we define its \emph{lattice graph} $(U, \Edges(U))$, where
$$\Edges(U) = \bigset{\set{u, u + \delta} : \delta \in \set{\phi, \psi}, u \in U, u + \delta \in U}.$$
\end{definition}
Intuitively, we connect every $u \in U \cap \Z^2$ with its translations by $\phi, \psi, -\phi, -\psi$ if they are contained in $U$.

\begin{lemma}\label{lattice_graph_connectivity}
	If $U$ is a truncated subtile, then $(U, \Edges(U))$ is connected.
	\begin{proof}
		Assume the contrary.
		Consider $u \in U$ and $v = u + s \phi + t \psi \in U$ such that
		\begin{itemize}
			\item $u$ and $v$ belong to different connected components of $(U, \Edges(U))$,
			\item $\absolute{s} + \absolute{t}$ is minimized.
		\end{itemize}
		W.l.o.g.~$s \ge 0$, since in the other case $u$ and $v$ can be swapped.
		We now show that there exists $w \in U$, such that $\set{u, w} \in \Edges(U)$ and if we let $s', t' \in \Z$ be such that $v = w + s'\phi + t'\psi$, then $\absolute{s'} + \absolute{t'} < \absolute{s} + \absolute{t}$, which contradicts the minimality of $\absolute{s} + \absolute{t}$. 
	
		Let $x_0, x_1, y_0, y_1, \phi_0, \phi_1, \psi_0, \psi_1 \in \Z$ and $\gamma \in \Z^2$ be the signature of $U$, meaning that
		\begin{itemize}
			\item $ U = [x_0, x_1] \times [y_0, y_1] \cap \set{z : z \in \Z^2, \h{z} \in [\phi_0, \phi_1], \s{z} \in [\psi_0, \psi_1], \gamma \equiv u}$,
			\item $x_1 - x_0 + 1 \ge \absolute{\x{\phi}} + \absolute{\x{\psi}}$ and $y_1 - y_0 + 1 \ge \absolute{\y{\phi}} + \absolute{\y{\psi}}$.
		\end{itemize}
		
		Recall that $\x{\phi} \ge 0$, $\y{\phi} \le 0$, $\x{\psi} \ge 0$, $\y{\psi} \ge 0$. We have the following cases:
		\begin{enumerate}[(a)]
			\item \label{it:s0_tneg} If $s = 0$, then w.l.o.g.~(up to swapping $u$ and $v$) $t > 0$. Let $w := u + \psi$. Observe that
				\eq{
					\x{u} \le \x{w} \le \x{v}, \quad \y{u} \le \y{w} \le \y{v}, \quad \h{u} \le \h{w} \le \h{v}, \quad \s{w} = \s{u},
				}
				and since $w \equiv u$, we get $w \in U$. 
			
			\item If $s > 0$ and $t = 0$, let $w := u + \phi$. Analogously, 
				\eq{
					\x{u} \le \x{w} \le \x{v}, \quad \y{u} \le \y{w} \le \y{v}, \quad \h{u} \le \h{w} \le \h{v}, \quad \s{w} = \s{u},
				}
				and since $w \equiv u$, we get $w \in U$. 
			
			\item If $s > 0$ and $t > 0$, let $w' := u + \phi$. \label{case_pos}
				We have \eq{
					\x{u} \le \x{w'} \le \x{v}, \quad \h{w'} = \h{u}, \quad \s{v} \le \s{w'} \le \s{u}.
				}
				If $w' \in U$, let $w = w'$.
				If $w' \not \in U$, then since all other requirements are satisfied, we must have $\y{w'} \not \in [y_0, y_1]$.
				Since $\y{w'} = \y{u} + \y{\phi} \le \y{u}$, we have $\y{u} + \y{\phi} \le y_0 - 1$, and
				considering $y_1 - y_0 + 1 \ge \absolute{\y{\phi}} + \absolute{\y{\psi}}$, we get $y_1 \ge \y{u} + \y{\psi}$.
				Now let $w := u + \psi$.
				We have \eq{
					\x{u} \le \x{w} \le \x{v}, \quad \y{u} \le \y{w} = \y{u} + \y{\psi} \le y_1, \quad \h{u} \le \h{w} \le \h{v}, \quad \s{w} = \s{u},
				}
				thus $w \in U$.
			\item If $s > 0$ and $t < 0$, consider $w' = u + \phi$.
				We have \eq{
					\y{v} \le \y{w'} \le \y{u}, \quad \h{w'} = \h{u}, \quad \s{v} \le \s{w'} \le \s{u}.
				}
				If $w' \in U$, then $w = w'$.
				Otherwise we can (similarly to (\ref{case_pos})) show that $w = u - \psi \in U$, by the fact that $x_1 - x_0 + 1 \ge \absolute{\x{\phi}} + \absolute{\x{\psi}}$. 
		\end{enumerate}
		In all these cases, which are clearly exhaustive, either $\absolute{s} \ge \absolute{s'} = 0$ and $\absolute{t} > \absolute{t'} = 1$ or $\absolute{s} > \absolute{s'} = 1$ and $\absolute{t} \ge \absolute{t'} = 0$ and hence $\absolute{s'}+\absolute{t'} \le \absolute{s}+\absolute{t}$, a contradiction. 
	\end{proof}
\end{lemma}


\begin{lemma}\label{monochromacy_condition}
	A truncated subtile string $S$ is monochromatic iff
	\[\Ham(S + \phi, S) + \Ham(S + \psi, S) = 0.\]
	\begin{proof}
		If $S$ is monochromatic, then clearly $\Ham(S + \phi, S) + \Ham(S + \psi, S) = 0$.
		Assume the contrary for the other implication.
		Let $u, v \in \d{S}$ be such that $S(u) \neq S(v)$.
		Since $\d{S}$ is a truncated subtile, the graph $(\d{S}, \Edges(\d{S})$ is connected by \Cref{lattice_graph_connectivity} and there must exist a path between $u$ and $v$.
		On that path there must exist a pair of neighbors $w, w'$, such that $S(w) \neq S(w')$ and $w' = w + \delta$ for some $\delta \in \set{\phi, \psi}$.
		If $\delta = \phi$, then $\Ham(S + \phi, S) \ge 1$ and if $\delta = \psi$, then $\Ham(S + \psi, S) \ge 1$ and we get a contradiction.
	\end{proof}
\end{lemma}

\begin{lemma}\label{cut_partitioning}
Given a (truncated) subtile string $S$, there is an algorithm that runs in time $\tO(\absolute{\d{S}} + 1)$ and constructs the following two partitionings.
	\begin{enumerate}[(a)]
		\item A partitioning of $S$ into a set of $\O(\Ham(S + \phi, S) + 1)$ (truncated) subtile strings $\U$, such that $\Ham(U + \phi, U) = 0$ for each $U \in \U$ and \label{partition_a}
		\item A partitioning of $S$ into a set of $\O(\Ham(S + \psi, S) + 1)$ (truncated) subtile strings $\V$, such that $\Ham(V + \psi, V) = 0$ for each $V \in \V$. \label{partition_b}
	\end{enumerate}
	\begin{proof}
		First, consider option (\ref{partition_a}). We construct the set
		\[ A = \set{\s{u} : u \in \d{S}, u + \phi \in \d{S}, S(u) \neq S(u + \phi)} \cup \set{-\infty, +\infty}\]
		and then sort its elements increasingly, creating a sequence $a_0, \dots, a_\ell$.
		Note that $\ell \le \Ham(S + \phi, S) + 2$.
		We then construct the strings $S_0, \dots, S_{\ell - 1}$, where $S_i$ is the restriction of $S$ to
		$\set{u : u \in \d{S}, \s{u} \in [a_i, a_{i + 1})}$ for every $i \in [\ell]$.
		Observe that $S_0, \dots, S_{\ell - 1}$ partition $S$ and that $\Ham(S_i + \phi, S_i) = 0$ for every $i \in [\ell]$. Finally, if $S$ is a (truncated) subtile string, then the created strings are (truncated) subtile strings as well.
	
		Second, consider option (\ref{partition_b}). Similarly to above, we construct
		\[ A = \set{\h{u} : u \in \d{S}, u + \psi \in \d{S}, S(u) \neq S(u + \psi)} \cup \set{-\infty, +\infty} \]
		and then sort it increasingly, creating $a_0, \dots, a_\ell$, where $\ell \le \Ham(S + \psi, S) + 2$.
		We then construct the strings $S_0, \dots, S_{\ell - 1}$, where $S_i$ is the restriction of $S$ to 
		$\set{u : u \in \d{S}, \h{u} \in (a_i, a_{i + 1}]}$. Similarly to above, if $S$ is a (truncated) subtile string, then the created strings are (truncated) subtile strings as well.
	\end{proof}
\end{lemma}

\restateThmTileDecomposition
\begin{proof}
We first construct in $\O(m^2)$ time a set $\Gamma \subseteq \Z^2$ such that $\absolute{\Gamma} = \O(\min\set{m, k})$ and $\forall u \in \Z^2 \; \exists \gamma \in \Gamma : u \equiv \gamma$ (\Cref{lattice_base}). Next, we use it to partition $R$ into a set of (truncated) subtile strings $\S$, such that $\absolute{\S} = \O(\min\set{m, k})$: Namely, for each $\gamma \in \Gamma$ (see \Cref{lattice_base}), we construct a (truncated) subtile string defined as the restriction of $R$ to $\d{R} \cap (\L + \gamma)$. For that, we consider each $u \in \d{R}$, construct its basis decomposition $u = s\cdot \phi + t\cdot \psi$, where $s,t\in \R$ by solving a system of linear equations in $\O(1)$ time, and add $u$ into the restriction corresponding to $\gamma = (s-\floor{s}) \cdot \phi + (t-\floor{t}) \cdot \psi \in \Gamma$. 

	We now apply \Cref{cut_partitioning} (\ref{partition_a}) to partition each $S \in \S$  into a set of (truncated) subtile strings $\S'$, such that $\S'$ partitions $R$ and $\Ham(S' + \phi, S') = 0$ for every $S' \in \S'$.
	Note that
	\[ \absolute{\S'} = \sum_{S \in \S} \O(\Ham(S + \phi, S) + 1) = \O(\Ham(R + \phi, R) + \absolute{\S}) = \O(k),\]
	since $\phi$ is a $\O(k)$-period of $R$.
	Finally, we apply  \Cref{cut_partitioning} (\ref{partition_b}) to partition each $S' \in \S'$ into a set of (truncated) subtile strings $\S''$, such that $\S''$ partitions $R$ and $\Ham(S'' + \phi, S'') = 0$ and $\Ham(S'' + \psi, S'') = 0$ for every $S'' \in \S''$.
	Again we have
	\[ \absolute{\S''} = \sum_{S' \in \S'} \O(\Ham(S' + \psi, S') + 1) = \O(\Ham(R + \psi, R) + \absolute{\S'}) = \O(k),\]
	since $R$ has a $\O(k)$-period $\psi$.
	By \Cref{monochromacy_condition}, the strings $S'' \in \S''$ are monochromatic.
	To construct~$\S''$, the algorithm uses $\tO(\absolute{\d{R}} + k)$ total time.
\end{proof}

Since $\absolute{\x{\phi}}, \absolute{\y{\phi}}, \absolute{\x{\psi}}, \absolute{\y{\psi}} \le n - m \le m / 2$, the string $P$ is a truncated tile string and we can apply \Cref{tile_decomposition} to partition it in time $\tO(m^2 + k)$ into a set of monochromatic truncated subtile strings $\V$. We then group the strings in $\V$ based on the single character they contain. Specifically, for every character $a \in \Sigma$ present in $P$, we construct a set $\V_a = \set{V : V \in \V, \getchar{V} = a}$.

\subsection{Text partitioning.}

While \cref{tile_decomposition} is sufficient to partition $P$ (which has $\O(k)$-periods $\phi$ and $\psi$) into a small number of monochromatic truncated subtile strings, $T$ is not necessarily periodic. Hence we are obliged to extend \cref{tile_decomposition} further to obtain a partitioning for $T$. We focus on the \emph{active text}, i.e., the positions of $T$ that actually participate in at least one of the $\O(k)$-mismatch occurrences of $P$.

\restateDefActiveText%
\restateObsActiveText

\restateDefPeripherality

In the remainder of the section, we show that $\Ta$ is close to being $\O(k)$-periodic with periods $\psi$ and $\phi$. 
This allows us to partition $\Ta$ into one peripheral string and a small number of monochromatic subtile strings.

\restateThmTextDecomposition

Now we prove the theorem. 
We will cut the text into parallelograms, and we assume that a parallelogram contains its border and its vertices.

\begin{definition}[Parallelogram grid]
	For $\ell \in \Z^+$, consider a sequence of $\ell+1$  distinct lines $h_0, \dots, h_\ell$ parallel to $\phi$ and a sequence of $\ell+1$ distinct lines $v_0, \dots, v_\ell$ parallel to $\psi$.
	For every $i, j \in [\ell + 1]$ let $w_{i, j}$ be the intersection of $h_i$ and $v_j$ (since $\phi$ and $\psi$ are not collinear, it is well-defined).
	For every $i, j \in [\ell]$ define $p_{i, j}$ to be the parallelogram (possibly trivial) with vertices $w_{i, j}, w_{i + 1, j}, w_{i + 1, j + 1}, w_{i, j + 1}$. We say that a set $\{p_{i,j} : i,j \in [\ell]\}$, is a \emph{parallelogram grid} of size $\ell$. 
\end{definition}

\begin{restatable}{lemma}{restateLemParallelogramGrid}\label{lm:parallelogram_grid}
For all $\ell \in \Z^2$, one can construct in $\O(m^2 + \ell^2)$ time a parallelogram grid $\{p_{i,j} : i,j \in [\ell]\}$, which satisfies each of the following:
\begin{enumerate}[(a)]
\item \label{it:unique} For all $u \in [n]^2$, there exist unique $i,j$ such that $u \in p_{i,j}$.

\item \label{it:small_parallelogram} For all $i, j \in [\ell]$ and every $u, v \in \X(p_{i, j}) \times \Y(p_{i, j})$, we have $\absolute{u - v} = \O(n / \ell)$.

\item \label{it:monotonicity} For all $i \in [\ell - 1]$ and $j \in [\ell]$ we have
	\eq{
		\min \X(p_{i, j}) &< \min \X(p_{i + 1, j}) & \min \Y(p_{i, j}) &\le \min \Y(p_{i + 1, j}) \\
		\max \X(p_{i, j}) &< \max \X(p_{i + 1, j}) & \max \Y(p_{i, j}) &\le \max \Y(p_{i + 1, j})
	}
	
For all $i \in [\ell]$ and $j \in [\ell - 1]$ we have
	\eq{
		\min \X(p_{i, j}) &\ge \min \X(p_{i, j + 1})  & \min \Y(p_{i, j}) &< \min \Y(p_{i, j + 1}) \\
		\max \X(p_{i, j}) &\ge \max \X(p_{i, j + 1})  & \max \Y(p_{i, j}) &< \max \Y(p_{i, j + 1})
	}
\end{enumerate}
\end{restatable}

\begin{figure}
        \centering
        \includegraphics[width=0.75\textwidth]{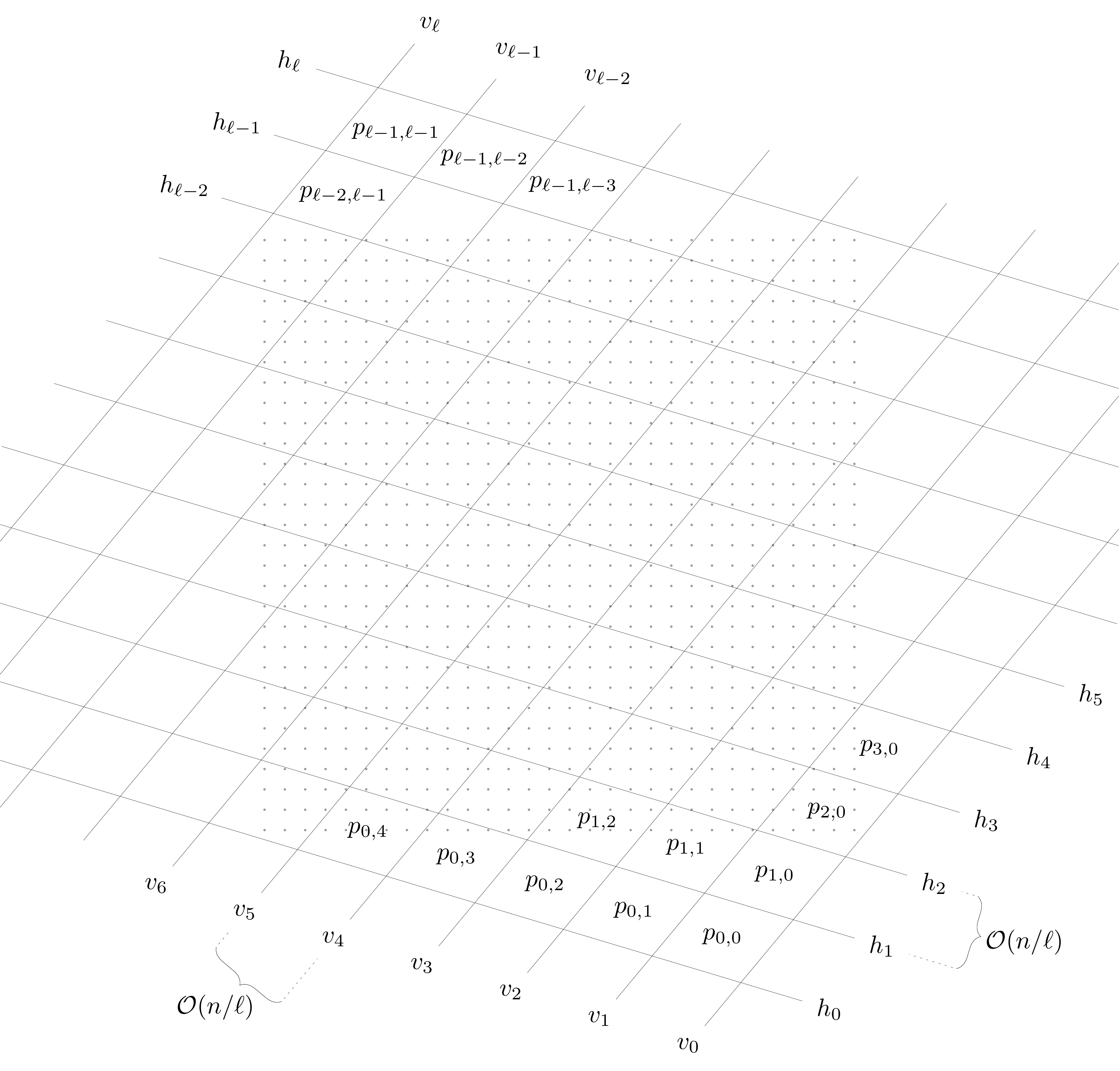}
        \caption{Parallelogram grid.}
        \label{figure:lemma5}
\end{figure}

Intuitively, it suffices to take a sequence of $\ell+1$ equispaced lines parallel to $\psi$ and a sequence of $\ell+1$ equispaced lines parallel to $\phi$ (see \cref{figure:lemma5}). The first property then follows from the fact that the angle between $\psi$ and $\varphi$ is in $[30^\circ,150^\circ]$ (hence, the parallelograms in the grid cannot be too ``stretched out''). The second property is implied by the fact that $\psi$ is in $(0,+\infty)\times [0,+\infty)$ and $\phi$ in $[0,\infty) \times (-\infty,0)$. A formal proof, which uses nothing but standard arithmetic, can be found in \cref{app:parallelogram_grid}.  

We apply \Cref{lm:parallelogram_grid} to construct a parallelogram grid $\{p_{i,j}\}$ of size $\ell$. If $\exists i, j \in [\ell] : \max \X(p_{i, j}) - \min \X(p_{i, j}) \ge m / 4$, then by \Cref{lm:parallelogram_grid} (\ref{it:small_parallelogram}) we have $m / 4 \le \max \X(p_{i, j}) - \min \X(p_{i, j}) = \O(n / \ell)$, and thus $\ell = \O(1)$. In this case, we can return a trivial partitioning where the set of monochromatic subtile strings is empty and $F = \Ta$, since $\Ta$ is $\O(m)$-peripheral. We can use an analogous argument if $\exists i,j \in \ell : \max \Y(p_{i, j}) - \min \Y(p_{i, j}) \ge m / 4$. Below we assume $\max \X(p_{i, j}) - \min \X(p_{i, j}) < m / 4$ and $\Y(p_{i, j}) - \min \Y(p_{i, j}) < m / 4$ for all $i, j \in [\ell]$. 

\begin{definition}
A set $U \subseteq \Z^2$ is \emph{coverable} if $U \subseteq \d{P + q}$ for some $q \in Q$.
\end{definition}

The following claim gives the intuition behind this definition:

\begin{claim}\label{coverable is periodic}
	The restriction of $T$ to a coverable set has $\O(k)$-periods $\phi$ and $\psi$.
\end{claim}	
	\begin{proof}
		Let $R$ denote the restriction. For $q \in Q$, such that $\d{R} \subseteq \d{P + q}$ we have
\begin{align*}
&\Ham(R + \phi, R) \le \\
&\Ham(R + \phi, P + q + \phi) + \Ham(P + q + \phi, P + q) + \Ham(P + q, R) \le \\
&\Ham(T, P + q) + \Ham(P + \phi, P) + \Ham(P + q, T) = \O(k)
\end{align*}
and identically $\Ham(R + \psi, R) = \O(k)$.
\end{proof}

Let $z = \frac{n - 1}{2}$. Split the 2D plane with two lines $x = z$ and $y = z$ into four quarters:
\begin{align}
\label{eq:quarters}
\begin{split}
&K_1 = (z, +\infty) \times (z, +\infty) \quad\quad K_2 = (-\infty, z) \times (z, +\infty)\\
&K_3 = (-\infty, z) \times (-\infty, z) \quad\quad K_4 = (z, +\infty) \times (-\infty, z)
\end{split}
\end{align}

\begin{claim}\label{coverable}
Let $P$ be a union of parallelograms~$p_{i,j}$ which have a non-empty intersection with $K_1$. If $v = (\max \X(P), \max \Y(P)) \in \d{\Ta}$, then $U = P \cap \Z^2$ is coverable. (For $K_2$, an analogous claim is true with $v: = (\min \X(P), \max \Y(P))$, for $K_3$ with $v : = (\min \X(P), \min \Y(P))$, and for $K_4$ with $v : = (\max \X(P), \min \Y(P))$.)
\end{claim}
\begin{proof}
As $v \in \d{\Ta}$, there exists $q \in Q$ such that $v \in [m]^2 + q$. As $P$ is a union of parallelograms $p_{i,j}$ which have a non-empty intersection with $K_1$, we have $\x{v} \ge z$ and $\y{v} \ge z$. Recall that we assume $n \le \frac{3}{2} m$ and $\max \X(p_{i, j}) - \min \X(p_{i, j}) < m / 4$ and $\Y(p_{i, j}) - \min \Y(p_{i, j}) < m / 4$ for all $i, j \in [\ell]$. Therefore, for all $u \in P$ we have
\begin{align*}
&\x{q} \le n-m \le n/2-m/4 \le \x{u} \le \x{v} \le \x{q} + m-1\\
&\y{q} \le n-m \le n/2-m/4 \le \y{u} \le \y{v} \le \y{q} + m-1 
\end{align*}
\noindent Consequently, $P \subseteq [m]^2+q$ and hence coverable.
\end{proof}

\begin{claim}\label{peripheral}
For all $i,j \in [\ell]$, the set $U = p_{i,j} \cap \Z^2$ is $\O(n / \ell)$-peripheral if one of the points $(\min \X(U), \min \Y(U))$, $(\min \X(U), \max \Y(U))$, $(\max \X(U), \min \Y(U))$, $(\max \X(U), \max \Y(U))$ is not in $\d{\Ta}$.
\end{claim}
\begin{proof}
W.l.o.g., $v = (\max \X(U), \max \Y(U)) \notin \d{\Ta}$. For all $u \in U$, we have $\absolute{u-v} = \O(n / \ell)$ by \cref{lm:parallelogram_grid} (\ref{it:small_parallelogram}), and hence $U$ is $\O(n / \ell)$-peripheral.
\end{proof}

\begin{lemma}\label{lm:parallelogram_merge}
Given $\ell \in \Z^+$, $\d{\Ta}$ can be covered with a set of $\O(\ell)$ parallelograms $\mathcal{P}$ with sides collinear to $\phi$ and $\psi$ such that for every $p \in \mathcal{P}$ the set $p \cap \Z^2$ is either coverable or $\O(n / \ell)$-peripheral. 
\end{lemma}
\begin{proof}
By the definition of $\Ta$, there does not exist $u \in \Z^2 \cap K_1 \setminus \d{\Ta}$ and $v \in \d{\Ta}$, such that $\x{u} \le \x{v}$ and $\y{u} \le \y{v}$. From this, \Cref{lm:parallelogram_grid}(\ref{it:monotonicity}), \cref{coverable}, and \cref{peripheral} for every $j \in [\ell]$ there exist $s,s', t',t \in [\ell]$  such that each of the following is satisfied:
\begin{enumerate}
	\item A parallelogram $P_j^1 = \bigcup_{i=0}^{s-1} p_{i, j} \cup \Z^2$ is $\O(n/\ell)$-peripheral,
	\item A parallelogram $P_j^2 = \bigcup_{i=s}^{s'-1} p_{i, j} \cup \Z^2$ is coverable, 
	\item Parallelograms $p_{s, j}, \dots, p_{t - 1, j}$ are not fully contained neither in $K_1$ nor in $K_3$,
	\item A parallelogram $P_j^3 = \bigcup_{i=t}^{t'-1} p_{i, j} \cup \Z^2$ is coverable,
	\item A parallelogram $P_j^4 = \bigcup_{i=t'}^{\ell - 1} p_{i, j} \cup \Z^2$ is $\O(n/\ell)$-peripheral.
\end{enumerate}

\noindent Similarly, for every $i \in [\ell]$ there exist $s,s', t',t \in [\ell]$  such that each of the following is satisfied:
\begin{enumerate}
	\item A parallelogram $R_i^1 = \bigcup_{j=0}^{s-1} p_{i, j} \cup \Z^2$ is $\O(n/\ell)$-peripheral,
	\item A parallelogram $R_i^2 = \bigcup_{j=s}^{s'-1} p_{i, j} \cup \Z^2$ is coverable, 
	\item Parallelograms $p_{i, s}, \dots, p_{i, t-1}$ are not fully contained neither in $K_2$ nor in $K_4$,
	\item A parallelogram $R_i^3 = \bigcup_{j=t}^{t'-1} p_{i, j} \cup \Z^2$ is coverable,
	\item A parallelogram $R_i^4 = \bigcup_{j=t'}^{\ell - 1} p_{i, j} \cup \Z^2$ is $\O(n/\ell)$-peripheral. 
\end{enumerate}

\begin{figure}
        \centering
        \includegraphics[width=0.75\textwidth]{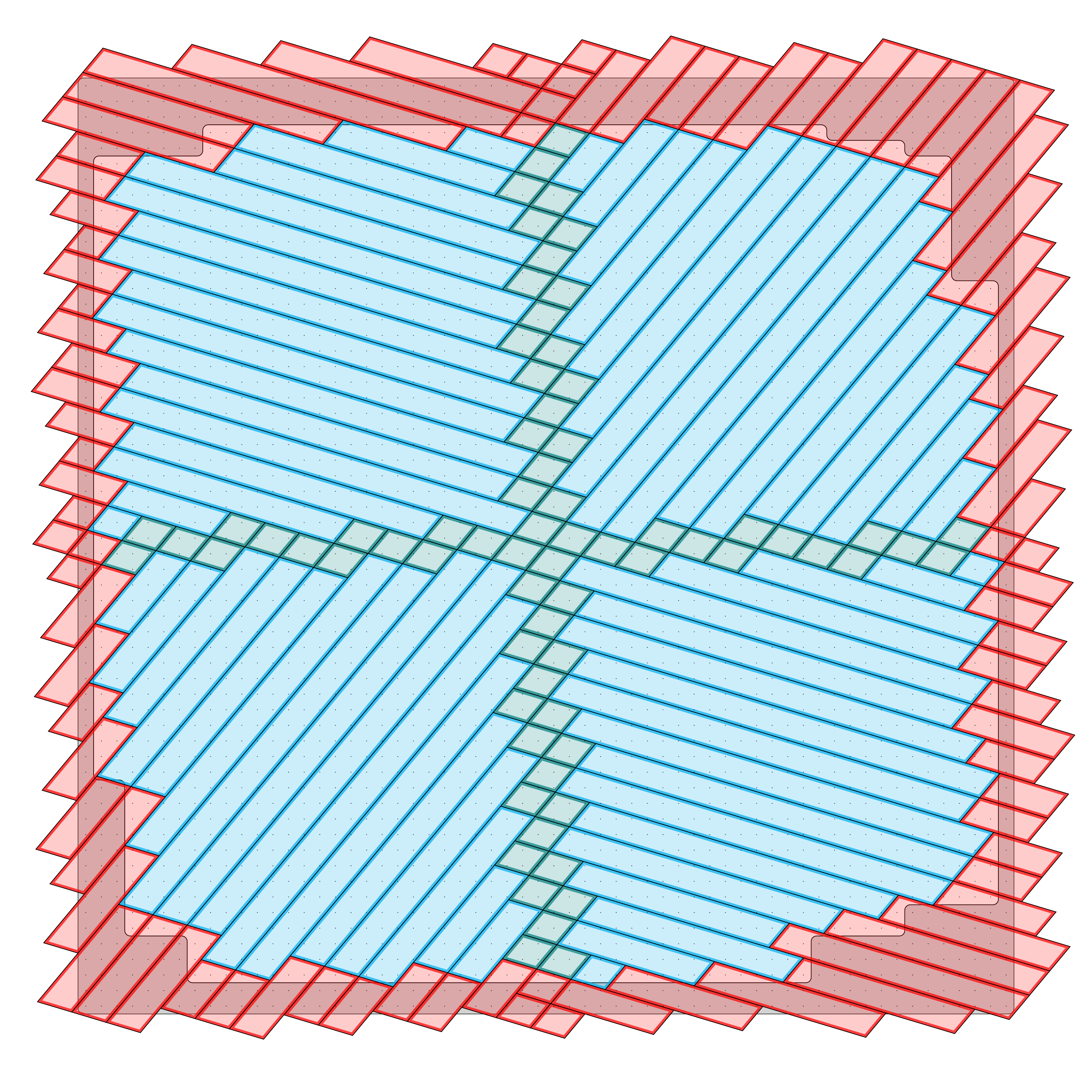}
        \caption{Parallelogram cover of $\Ta$. The polygon in the center corresponds to $\Ta$. Red parallelograms are $\O(n/\ell)$-peripheral, blue coverable, and green belong to $\I$.}
        \label{figure:text_decomposition}
\end{figure}

Note that there are $\O(\ell)$ parallelograms $P_j^x$ and $R_i^x$ for $i,j \in [\ell]$ and $x \in \set{1, \dots, 4}$ and that they cover all parallelograms $p_{i,j}$ but those that are not fully contained in any of the quarters, in other words, the parallelograms $p_{i, j}$ intersecting  at least one of the lines $x = z$ and $y = z$. Denote the set of such parallelograms by $\I$. 

By \cref{coverable} and \cref{peripheral}, for each parallelogram $p_{i,j} \in \I$ the set $p_{i,j} \cap \Z^2$ is either coverable or $\O(n/\ell)$-peripheral. Furthermore, $\absolute{\I} = \O(\ell)$. Consider a set $U_x$ of intersections of grid lines different from $x=z$ with the line and a set $U_y$ of intersections of grid lines different from $y=z$. Every parallelogram in $\I$ contains $u \in U_x \cup U_y$, and every $u \in \Z^2$ can be contained in at most four parallelograms, implying the bound. 

The set $\mathcal{P} = \I \cup \{P_j^x : j \in [\ell], x \in \{1, \dots, 4\}\} \cup \{R_j^x : j \in [\ell], x \in \{1, \dots, 4\}\}$ satisfies the claim of the lemma. (See \cref{figure:text_decomposition}.) It remains to explain how to construct it efficiently. For all $i,j \in [\ell]$, the values $\min \X(p_{i,j}), \max \X(p_{i,j}), \min \Y(p_{i,j}), \max \X(p_{i,j})$ can be inferred in $\O(1)$ time. As $p_{i,j}$ is a connected figure, it intersects with $x = z$ if and only if $\min \X(p_{i,j}) \le \max \X(p_{i,j})$ and with $y = z$ if and only if $\min \Y(p_{i,j}) \le \max \Y(p_{i,j})$. Both of the conditions can be checked in $\O(1)$ time as well. Plugging this bounds into our construction, we obtain that $\mathcal{P}$ can be built in $\O(\ell^2) = \O(m^2)$ time. 
\end{proof}

We apply \cref{lm:parallelogram_merge} to $\Ta$. Let $\mathcal{P}_1 = \{p \in \mathcal{P} : p \text { is coverable}\}$ and $\mathcal{P}_2 = \{p \in \mathcal{P} : p \text { is }\O(n/\ell)\text{-peripheral}\}$. For each $p \in \mathcal{P}_1$, consider a restriction $R$ of $\Ta$ to $p$. By \Cref{coverable is periodic}, it has $\O(k)$-periods $\phi$ and $\psi$, thus, by \Cref{tile_decomposition}, it can be partitioned in time $\tO(\absolute{\d{R}} + k)$ (after $\O(m^2)$-time shared preprocessing of $\phi$ and $\psi$) into $\O(k)$ monochromatic subtile strings. Finally, construct in $\O(m^2)$ time the restriction of $\Ta$ to $\cup_{p \in \mathcal{P}_2} p$, which is a $\O(n / \ell)$-peripheral string. The total time of the partitioning algorithm is therefore $\O(m^2+\ell k)$.

This concludes the proof of \cref{text_decomposition}.

\section{Algorithm}
\label{sec:algorithm}

Now we show how to algorithmically exploit the partitioning of the pattern and the text. In \cref{sec:algorithm:center}, we show that the Hamming distances between $P$ and a string that can be partitioned into a set of monochromatic subtile strings can be computed efficiently:

\restateThmSparseAlgo

As an immediate corollary, we obtain an $\tO(m^2 + mk^2)$-time solution for the $k$-mismatch problem. First, we apply \Cref{text_decomposition} for a large enough value $\ell = \Theta(m)$ to partition the active text $\Ta$ into a set $\S$ of $\O(mk)$ monochromatic subtile strings and a $0$-peripheral string, empty by definition, in time $\tO(m^2 + mk)$. For every $q \in Q$ we then have $\Ham(P + q, \Ta) = \sum_{S \in \S} \Ham(P + q, S)$, and these values can be computed in $\tO(m^2 + mk^{2})$ total time by \cref{th:sparse_algo}, since $\sum_{S \in \S} \absolute{\V_{\getchar{S}}} \le \absolute{\S} \absolute{\V} = \O(mk \cdot k)$. 

To improve the complexity further, we partition the active text using the algorithm from \Cref{text_decomposition} with $\ell = mk^{-3/4}$.
We obtain a set~$\S$ of $\O(mk^{1/4})$ monochromatic subtile strings $\S$, and an $\O(k^{3 / 4})$-peripheral string $F$. For every $q \in Q$ we then have

\[ \Ham(P + q, \Ta) = \Ham(P + q, F) + \sum_{S \in \S} \Ham(P + q, S).\]

By \cref{th:sparse_algo}, we can compute $\sum_{S \in \S} \Ham(P + q, S)$ for every $q \in Q$ in time $\tO(m^2 + mk^{5/4})$, since similarly to above we have $\sum_{S \in \S} \absolute{\V_{\getchar{S}}} \le \absolute{\S} \absolute{\V} = \O(mk^{5/4})$. It remains to compute the values $\Ham(P + q, F)$. In \cref{sec:algorithm:periph}, we show how to efficiently compute the Hamming distances between $P$ and a $d$-peripheral string:

\restateThmDenseAlgo

By substituting $d = \O(k^{3/4})$, we obtain the desired time complexity of $\tO(m^2 + mk^{5/4})$ and hence complete the proof of \cref{th:main}.

\subsection[Counting mismatches between P and monochromatic subtile strings.]{Proof of \cref{th:sparse_algo}.}\label{sec:algorithm:center}
\restateThmSparseAlgo*
Let $U = \bigcup_{S \in \S} \d{S}$. Recall that $\V$ is a set of monochromatic truncated subtile strings that partition $P$, and $\V_a = \set{V : V \in \V, \getchar{V} = a}$ (see \cref{sec:partition:pattern}).
 Observe that
\begin{align}
\label{eq:sparse}
\sum_{S \in \S} \Ham(P + q, S) = \absolute{\d{P + q} \cap U} - \sum_{S \in \S}\sum_{V \in \V_{\getchar{S}}} \absolute{\d{S} \cap \d{V + q}}.
\end{align}

To compute $\absolute{\d{P + q} \cap U}$ for all $q \in Q$, we create strings $P'$ and $T'$ with respective domains $\dom(P)$ and $\dom(T)$. Both strings are filled with character $0$. 
In $T'$, we then replace all positions from $\dom(U)$ with character $1$. 
We observe that $\absolute{\d{P + q} \cap U} = \Ham(P' + q, T')$, which can be computed for all $q \in Q$ in $\tO(m^2)$ time with \cref{cor:sigman2d}.
It remains to show how to compute the second term, let it be $\tau_q$. 

Denote $D = \set{u : u \in \L, \h{u} \ge 0, \s{u} \le 0} = \{ s\phi + t\psi : s,t \in \mathbb N_{\geq 0}\}$, where $\L = \{ s\phi + t\psi : s,t \in \mathbb Z\}$ is the set from \Cref{lattice_congruency} (geometrically, $D$ is the part of the lattice spanned by the angle with the sides collinear to $\phi$ and~$\psi$).

\begin{lemma}\label{primitive}
Let $S$ be a subtile with signature $\phi_0, \phi_1, \psi_0, \psi_1, \gamma$. In constant time, we can compute points $w_{0,0},w_{0,1},w_{1,0},w_{1,1} \in \mathbb Z^2$ such that for every $X \in \mathbb Z^2$ it holds
\[ \absolute{S \cap X} = \sum_{i,j \in \{0,1\}} (-1)^{(i+j)} \cdot \absolute{(D + w_{i,j}) \cap X}\]
and $\forall u \in (D + w_{i,j}) : u \equiv \gamma$. If $\max_{u \in S}(\absolute{u}) = \O(m)$, then $\absolute{w_{i,j}} = \O(m)$.

\begin{proof}
By definition of a subtile, every point $u \in S$ satisfies $u - \gamma = s\phi + t\psi$, or equivalently $u = s\phi + t\psi + \gamma$ for some $s,t\in\mathbb Z$. 
By definition of a tile, it also holds $\phi \times (s\phi + t\psi + \gamma) \in [\phi_0, \phi_1]$ and $\psi \times (s\phi + t\psi + \gamma) \in [\psi_0, \psi_1]$. Since the cross-product has the distributive property, and the cross-product of collinear vectors is $0$, we can equivalently write $t \cdot (\phi \times \psi) + (\phi \times \gamma) \in [\phi_0, \phi_1]$ and $s \cdot (\psi \times \phi) + (\psi \times \gamma) \in [\psi_0, \psi_1]$. 
Due to $\psi \in (0, +\infty) \times [0, +\infty)$ and $\phi \in [0, +\infty) \times (-\infty, 0)$, it holds $(\phi \times \psi) \geq 0$ and $(\psi \times \phi) \leq 0$. Hence the former term increases with $t$, while the latter one decreases as $s$ increases.

By solving the respective univariate linear equations in constant time, we determine the respectively minimal $t_0, t_1 \in \mathbb Z$ and maximal $s_0, s_1 \in \mathbb Z$ such that $\psi \times (s_0\phi + \gamma) \leq \psi_1$, $\psi \times (s_1\phi + \gamma) < \psi_0$, $\phi \times (t_0\psi + \gamma) \geq \phi_0$ and $\phi \times (t_1\psi + \gamma) > \phi_1$.
For $i,j \in \{0, 1\}$, we define $w_{i, j} = (s_i\phi + t_j\psi + \gamma) \equiv \gamma$ and $D_{i,j} = D + w_{i,j}$. 
Then, a point $s\phi + t\psi + \gamma$ with $s,t \in \mathbb Z$ is in 
\begin{itemize}
	\item $S$ if and only if $s_0 \leq s < s_1$ and $t_0 \leq t < t$, 
	\item $D_{0, 0}$ if and only if $s_0 \leq s$ and $t_0 \leq t$, 
	\item $D_{0, 1}$ if and only if $s_0 \leq s$ and $t_1 \leq t$, 
	\item $D_{1, 0}$ if and only if $s_1 \leq s$ and $t_0 \leq t$, 
	\item $D_{1, 1}$ if and only if $s_1 \leq s$ and $t_1 \leq t$.
\end{itemize}

\vspace{.5\baselineskip}

\parbox{.55\textwidth}{\begin{tikzpicture}[x=.6em, y=.6em]

\foreach \s in {1,...,10} {
	\foreach[evaluate=\s as \x using (\s+.7*\t),
			 evaluate=\s as \y using (0.8*\s-.5*\t + 5)] \t in {1,...,10} {
		\node[black!20!white, draw, circle, fill=black!20!white, inner sep=.5pt] (\s-\t) at (\x, \y) {};	
	}
}

\foreach \s in {3,...,8} {
	\foreach \t in {3,...,8} {
		\node[draw, circle, fill, inner sep=.5pt] at (\s-\t) {};	
	}
}

\draw ($0.5*(2-1)+0.5*(3-1)$) to node[pos=0, above left] {\small $\phi_0$} ($0.5*(2-10)+0.5*(3-10)$);
\draw ($0.5*(8-1)+0.5*(9-1)$) to node[pos=0, above left] {\small $\phi_1$} ($0.5*(8-10)+0.5*(9-10)$);
\draw ($0.5*(1-2)+0.5*(1-3)$) to node[pos=0, below left] {\small $\psi_1$} ($0.5*(10-2)+0.5*(10-3)$);
\draw ($0.5*(1-8)+0.5*(1-9)$) to node[pos=0, below left] {\small $\psi_0$} ($0.5*(10-8)+0.5*(10-9)$);

\node[draw, circle, fill=red, inner sep=1.5pt] (3-3) at (3-3) {};
\node[draw, circle, fill=red, inner sep=1.5pt] (3-9) at (3-9) {};
\node[draw, circle, fill=red, inner sep=1.5pt] (9-3) at (9-3) {};
\node[draw, circle, fill=red, inner sep=1.5pt] (9-9) at (9-9) {};

\draw[densely dotted, thick] (3-3) ++(-1.5, 3.5) node[black, inner sep=1pt] (l) {\small $w_{0,0}$} (3-3) to[bend right] (l);
\draw[densely dotted, thick] (3-9) ++(4, 0) node[black, inner sep=1pt] (l) {\small $w_{1,0}$} (3-9) to[bend right] (l);
\draw[densely dotted, thick] (9-3) ++(4.5, 0) node[black, inner sep=1pt] (l) {\small $w_{0,1}$} (9-3) to[bend left] (l);
\draw[densely dotted, thick] (9-9) ++(1, -3) node[black, inner sep=1pt] (l) {\small $w_{1,1}$} (9-9) to[bend left] (l);

\foreach \s in {1,...,20} {
	\foreach[evaluate=\s as \x using (0.5*\s+.35*\t+0.425 + 19),
			 evaluate=\s as \y using (0.4*\s-.25*\t+0.325 + 5)] \t in {1,...,20} {
		\node[circle, inner sep=.5pt] (\s-\t) at (\x, \y) {};	
	}
}

\foreach \s in {1,...,20} {
		\ifnum\s=8\else
			\draw[gray, thin] (\s-8.center) to (\s-20.center);
			\draw[gray, thin] (8-\s.center) to (20-\s.center);
		\fi
}

\draw[very thick]
(1-1.center) to (20-1.center)
(1-1.center) to (1-20.center)
(1-8.center) to (20-8.center)
(8-1.center) to (8-20.center);
\path (1-8.center) to node[midway] {$S$} (8-1.center);

\node[draw, circle, fill=red, inner sep=1.5pt] (1-1) at ($0.5*(1-1)+0.5*(2-2)$) {};
\node[draw, circle, fill=red, inner sep=1.5pt] (1-8) at ($0.5*(1-8)+0.5*(2-9)$) {};
\node[draw, circle, fill=red, inner sep=1.5pt] (8-1) at ($0.5*(8-1)+0.5*(9-2)$) {};
\node[draw, circle, fill=red, inner sep=1.5pt] (8-8) at ($0.5*(8-8)+0.5*(9-9)$) {};

\end{tikzpicture}}%
\hfill%
\parbox{.38\textwidth}{
$D_{0,0}\ = \enskip\rlap{\smash{\scalebox{.75}{\begin{tikzpicture}[baseline=0em]
	\node[draw=black, pattern=north east lines, pattern color=gray, rotate=-45, inner sep=0, minimum width=1em, minimum height=2em, anchor=south east] at (0, 0) {};
	\node[draw=black, pattern=north west lines, pattern color=gray, rotate=-45, inner sep=0, minimum width=1.5em, minimum height=1.2em, anchor=north west] at (0, 0) {};
	\node[draw=black, pattern=crosshatch, pattern color=gray, rotate=-45, inner sep=0, minimum width=1.5em, minimum height=2em, anchor=south west] at (0, 0) {};
	\node[draw=black, rotate=-45, inner sep=0, minimum width=1em, minimum height=1.2em, anchor=north east] at (0, 0) {};
\end{tikzpicture}}}}$%
~~~~~~~~~~~~~~~
$D_{0,1}\ = \enskip\rlap{\smash{\scalebox{.75}{\begin{tikzpicture}[baseline=0em]
	\node[draw=black, pattern=north east lines, pattern color=gray, rotate=-45, inner sep=0, minimum width=1em, minimum height=2em, anchor=south east] at (0, 0) {};
	\node[draw=black, pattern=crosshatch, pattern color=gray, rotate=-45, inner sep=0, minimum width=1.5em, minimum height=2em, anchor=south west] at (0, 0) {};
\end{tikzpicture}}}}$

\vspace{2\baselineskip}

$D_{1,0}\ = \enskip\rlap{\smash{\scalebox{.75}{\begin{tikzpicture}[baseline=0em]
	\node[draw=black, pattern=north west lines, pattern color=gray, rotate=-45, inner sep=0, minimum width=1.5em, minimum height=1.2em, anchor=north west] at (0, 0) {};
	\node[draw=black, pattern=crosshatch, pattern color=gray, rotate=-45, inner sep=0, minimum width=1.5em, minimum height=2em, anchor=south west] at (0, 0) {};
\end{tikzpicture}}}}$%
~~~~~~~~~~~~~~~
$D_{1,1}\ = \enskip\rlap{\smash{\scalebox{.75}{\begin{tikzpicture}[baseline=0em]
	\node[draw=black, pattern=crosshatch, pattern color=gray, rotate=-45, inner sep=0, minimum width=1.5em, minimum height=2em, anchor=south west] at (0, 0) {};
\end{tikzpicture}}}}$}

\vspace{.5\baselineskip}

All sets contain only points that are in the same class as $\gamma$. 
Since $D_{0,0}$ is partitioned by $S \cup (D_{0, 1}\setminus D_{1, 1}) \cup (D_{1, 0}\setminus D_{1, 1}) \cup D_{1, 1}$, the statement of the lemma follows. 
Any $w_{i, j}$ is at distance at most $\absolute{\phi + \psi} = \O(m)$ from a point in $S$, and $\max \{\absolute{u} : u \in S\} = \O(m)$ implies $\absolute{w_{i,j}} = \O(m)$. 
(If $S$ is empty, i.e., $s_0 = s_1$ or $t_0 = t_1$, then we define $\forall i,j\in\{0, 1\} : w_{i,j} = \gamma$ instead.)
\end{proof}

\end{lemma}

We apply \Cref{primitive} to the domain of each $S\in \S$ to obtain $w_{i, j}(S)$ for $i,j \in \{0,1\}$ in $\O(\absolute{\S})$ total time. We then have:

\begin{align}
\label{eq:S_replaced_by_shifted_angles}
\begin{split}
\tau_q &= \sum_{S \in \S} \sum_{V \in \V_{\getchar{S}}} \sum_{i,j \in \{0,1\}} (-1)^{i+j} \absolute{\d{D+w_{i, j}(S)} \cap \d{V + q}} \\
& = \sum_{i,j \in \{0,1\}} (-1)^{(i+j)} \sum_{S \in \S} \sum_{V \in \V_{\getchar{S}}} \absolute{\d{D-q} \cap \d{V-w_{i, j}(S)}}
\end{split}
\end{align}
Now, fix $i,j \in \{0,1\}$, and let 
\begin{equation}
\label{eq:fixed_ij}
\tau_q^{i,j} = \sum_{S \in \S} \sum_{V \in \V_{\getchar{S}}} \absolute{\d{D-q} \cap \d{V-w_{i, j}(S)}}
\end{equation}

Now we explain how to compute the values $\tau_q^{i,j}$ for all $q \in Q$ and all $i,j \in \{0,1\}$ in time $\tO(m^2 + \sum_{S \in \S} \absolute{\V_{\getchar{S}}})$. Then, the values $\tau_q$ can be computed in $\O(\absolute{Q}) = \O(m^2)$ time using $\tau_q = \tau^{0,0}_q - \tau^{0,1}_q - \tau^{1,0}_q + \tau^{1,1}_q$ due to \cref{eq:S_replaced_by_shifted_angles,eq:fixed_ij}. 
Define a set of truncated subtiles $\S \odot \V := \{V-w_{i, j}(S) : S \in \S, V \in \V_{\getchar{S}}\}$.  Note that $\absolute{\S \odot \V} = \O(\sum_{S \in \S} \absolute{\V_{\getchar{S}}})$, and for all $u \in V$ with $V \in \S \odot \V$, we have $\absolute{u} = \O(m)$. We can simplify \cref{eq:fixed_ij} as follows:

\begin{align}
\label{eq:group_by_gamma}
\tau_q^{i,j} = \sum_{V \in \S \odot \V} \absolute{\d{D-q} \cap \d{V}}
\end{align}

\begin{lemma}\label{primitive_conv}
The values $\tau_q^{i,j}$ for all $q \in Q$ can be computed in total time $\tO(m^2 + \sum_{S \in \S} \absolute{\V_{\getchar{S}}})$. 
\end{lemma}
\begin{proof}
For $u \in \Z^2$, let $\score(u)$ be the number of truncated subtiles in $\S \odot \V$ containing $u$, i.e., $\score(u) = \absolute{\{V : V \in \S \odot \V, u \in \dom(V)\}}$. Observe that
\begin{equation}
\begin{split}
\tau_q^{i,j} =&\sum_{V \in \S \odot \V} \absolute{\{u \in \Z^2 : u \in \dom(D-q), u \in \dom(V)\}}\\
=&\sum_{V \in \S \odot \V}\ \sum_{u \in \dom(D-q)}  \absolute{\{u : u \in \dom(V)\}}\\
=&\sum_{u \in \dom(D-q)}\ \sum_{V \in \S \odot \V} \absolute{\{u : u \in \dom(V)\}} = \sum_{u \in \dom(D-q)} \score(u).
\end{split}
\end{equation}
We now explain how to compute $\tau_q^{i,j}$. 
We start by computing an upper bound $\ell = \O(m)$ such that $\score(u) = 0$ for any $u$ with $\absolute{u} > \ell$. Note that $\score(u)>0$ requires that $u \in \dom(V - w_{i,j}(S))$ for some $V \in \V$ and $S \in \S$. Since $\dom(V) \subseteq \dom(P) = [m]^2$, we can use $\ell = \absolute{(m,m)} + \max_{S \in \S}(\absolute{w_{i, j}(S)})$, computable in $\O(\absolute{S})$ time, where $\ell = \O(m)$ due to $\absolute{w_{i,j}(S)} = \O(m)$.
%
Therefore, the set $\W' = \{(x,y) \in \Z^2 : \absolute{x} \le \ell, \absolute{y} \le \ell \}$ is of size $\O(m^2)$. 
We pad this square-shaped set such that it is a parallelogram with sides collinear to $\phi$ and $\psi$, i.e., 
$$\W = \{ w \in \mathbb Z^2 : \absolute{\phi \times w} \leq \absolute{\phi \times (\ell, \ell)}, \absolute{\psi \times w} \leq \absolute{\psi \times (-\ell, \ell)}$$

\begin{center}
\begin{tikzpicture}[x=.15em, y=.15em, remember picture]

\clip (-20,-20) rectangle (30,30);

\node (bl) at (0,0) {};
\node (tr) at (10,10) {};
\node[draw, densely dotted, thick, inner sep=-0.5pt, fit=(bl.center)(tr.center)] {};

\draw[very thick] (0,10) ++(-100,-80) to ++(200,160);
\draw[very thick] (10,0) ++(-100,-80) to ++(200,160);
\draw[very thick] (0,0) ++(-100,30) to ++(200,-60);
\draw[very thick] (10,10) ++(-100,30) to ++(200,-60);

\fill[white] (-100, 300) to (-100,-70) to ++(200,160) to (-100,300);
\fill[white] (-100, 300) to (-90,40) to ++(200,-60) to (-100,300);
\fill[white] (-100, -300) to (-90,-80) to ++(200,160) to (-100,-300);
\fill[white] (-100, -300) to (-100,30) to ++(200,-60) to (-100,-300);

\draw[thin, black!20!white] (0,10) ++(-100,-80) to ++(200,160);
\draw[thin, black!20!white] (10,0) ++(-100,-80) to ++(200,160);
\draw[thin, black!20!white] (0,0) ++(-100,30) to ++(200,-60);
\draw[thin, black!20!white] (10,10) ++(-100,30) to ++(200,-60);

\draw[densely dashed] (-20,-20) rectangle (30,30);

\node (root) at (5,5) {};

\end{tikzpicture}~\begin{tikzpicture}[x=.15em, y=.15em]
	\tikzset{every node/.style={inner sep=0}}
	\node (top) at (-10,5) {};
	\node (bot) at (0,-45) {};
	\node (l1) at (0,0) {};
	\node (l2) at (8,-8) {};
	\node (l3) at (0,-16) {};
	\node (l4) at (8,-24) {};
	\node (l5) at (0,-32) {};
	\node (l6) at (8,-40) {};
	
	\node[fit=(l1)(l2)] (L1) {};
	\node[draw, densely dotted, thick, fit=(l3)(l4)] (L2) {};
	\node[draw, densely dashed, fit=(l5)(l6)] (L3) {};
	
	\node[right=2 of L1] {$\strut = \W$};
	\node[right=2 of L2] {$\strut = \W'$};
	\node[right=2 of L3] {$\strut = \W''$};
	
	\node[fit=(L1.25)(L1.205), inner sep=2] (L1) {};	
	\draw (L1.10) to (L1.100) to (L1.190) to (L1.280) to (L1.10);
\end{tikzpicture}
\end{center}

Since the angle between $\phi$ and $\psi$ is between $30^\circ$ and $150^\circ$, it can be readily verified that $\W$ is a subset of the larger square $\W'' = \{(x,y) \in \Z^2 : \absolute{x} \le 3\ell, \absolute{y} \le 3\ell \}$, which is still of size $\O(m^2)$. Hence, by checking each point in $\W''$ separately, we can construct $\W$ in $\O(m^2)$ time. By design, it holds $\score(w)=0$ for $w \notin \W$. 

We then construct in $\O(m^2)$ time the set $\Gamma$ defined in \Cref{lattice_base} and for each $\gamma \in \Gamma$ define $\W_\gamma = \{w \in \W : w \equiv \gamma\}$. For $\gamma \in \Gamma$, let $(\S \odot \V) |_\gamma = \{V \in \S \odot \V : V \equiv \gamma\}$. 
By \cref{subtile_definition}, if $w \in \W_\gamma$ is contained in a truncated subtile $V \in \S \odot \V$, then $V \equiv \gamma$, and consequently $V \in (\S \odot \V) |_\gamma$. Therefore, for $w \in \W_\gamma$, we can refine the definition of score:
$$\score(w) = \sum_{V \in (\S \odot \V) |_\gamma} \absolute{\{u : u \in \dom(V)\}}$$
Consider a truncated subtile $V \in (\S \odot \V) |_\gamma$ with signature $x_0, x_1, y_0, y_1, \phi_0, \phi_1, \psi_0, \psi_1, \gamma$. For any $w = (\x{w}, \y{w}) \in \W_\gamma$, it holds $w \in V$ if and only if
\begin{equation}
\label{eq:range}
x_0 \le x \le x_1, y_0 \le y \le y_1, \phi_0 \le \h{w} \le \phi_1, \psi_0 \le \s{w} \le \psi_1
\end{equation}
(where we do not need to check whether $w \equiv \gamma$ as $w \in \W_\gamma$). We build in $\tO(\absolute{(\S \odot \V) |_\gamma})$ time an $8$-dimensional range tree~\cite{BENTLEY1979244} for the set of points 
$$\{(x_0, x_1, y_0, y_1, \phi_0, \phi_1, \psi_0, \psi_1) : \exists V \in (\S \odot \V) |_\gamma\text{\ with signature\ } x_0, x_1, y_0, y_1, \phi_0, \phi_1, \psi_0, \psi_1, \gamma\}.$$
For $w = (w.x, w.y) \in \W_\gamma$, let $w.\psi = \psi \times w$ and $w.\phi = \phi \times w$. For every such $w$, we perform in $\tO(1)$ time a counting query for the range 
$$
(-\infty, w.x] \times [w.x, \infty) \times 
(-\infty, w.y] \times [w.y, \infty) \times
(-\infty, w.\phi] \times [w.\phi, \infty) \times
(-\infty, w.\psi] \times [w.\psi, \infty).
$$
The result is exactly the number of sets $\V \in (\S \odot \V) |_\gamma$ that contain $w$, i.e., $\score(w)$. The total time for all $\W_\gamma$ (including the construction of the range tree) is $\tO(\absolute{\W_\gamma} + \absolute{(\S \odot \V) |_\gamma})$. Summing over all $\gamma$, the time for computing $\score(w)$ for all $w \in \W$ is $\O(m^2 + \absolute{(\S \odot \V)})$.

Now, note that the following recursive formula holds for all $q \in [m]^2 \supseteq Q$:
$$\sum_{u \in D-q} \score(u) = \score(-q) + \sum_{u \in D-q+\phi} \score(u) + \sum_{u \in D-q+\psi} \score(u) -\sum_{u \in D-q+\phi+\psi} \score(u).$$
Note that, if $-q \in \W$ but $-q + \psi \notin \W$, then $(\sum_{u \in D-q+\psi} \score(u)) = 0$. This is because $\W$ and $D -q + \psi$ cannot overlap if $-q + \psi \notin \W$, and all scores outside $\W$ are $0$. The same holds for $-q + \phi$ and $-q + \phi + \psi$. Hence we can stop the recursion as soon as we leave $\W$. By using dynamic programming, we can therefore compute $\sum_{u \in D-q} \score(u)$ for every $q \in\W \subseteq \Q$ in $\O(\absolute{W}) = \O(m^2)$ time. 
We can then extract the values $\tau^{i,j}_q$ for all $q \in Q$ in $\absolute{Q} = \O(m^2)$ time, concluding the proof. 
\end{proof}

\subsection[Counting mismatches between P and a d-peripheral string.]{Proof of \cref{th:dense_algo}.}
\label{sec:algorithm:periph}

In this section, we show how to compute the Hamming distances between $P$ and a $d$-peripheral string efficiently:

\restateThmDenseAlgo*

We partition $F$ into four roughly equal-sized $d$-peripheral strings $F_1, \dots, F_4$, where $F_i$ is the restriction of $F$ to $K_i$ defined as in \cref{eq:quarters}. (See \cref{figure:quarter_split}.) For all $q \in Q$, we have 
$$\Ham(P + q, F) = \sum_{i = 1}^4 \Ham(P + q, F_i).$$
Below, we focus on computing $\Ham(P + q, F_1)$ for all $q \in Q$. The other cases are symmetric. 

For a string $V$, define a value $h(V)$ that is an upper bound on $V$'s height as follows: If $V$ is empty, $h(V): = 0$, and otherwise $h(V) := \min \{h \in \Z^+: \forall u,v \in V, (\x{v}, \y{u} + h) \notin \d{\Ta}\}$.
We start with the following technical claim.

\begin{figure}
    \centering
    \begin{subfigure}[t]{0.45\textwidth}
        \centering
        
        \includegraphics[width=0.9\textwidth]{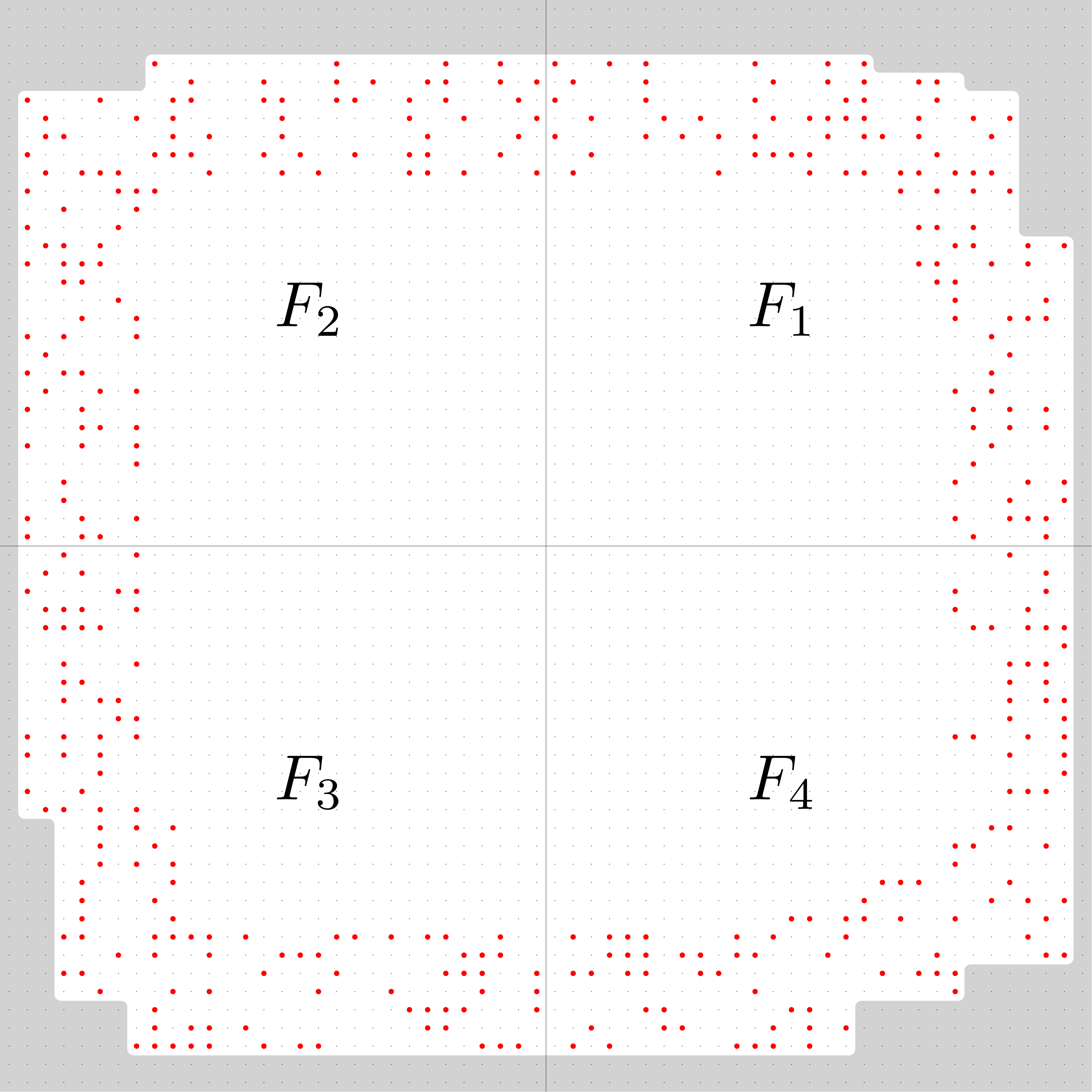}
        \caption{Strings $F_i$ in the partitioning of $F$. The red points constitute $\d{F}$.}
        \label{figure:quarter_split}
    \end{subfigure}%
    \hfill
    \begin{subfigure}[t]{0.45\textwidth}
        \centering
        \includegraphics[width=0.9\textwidth]{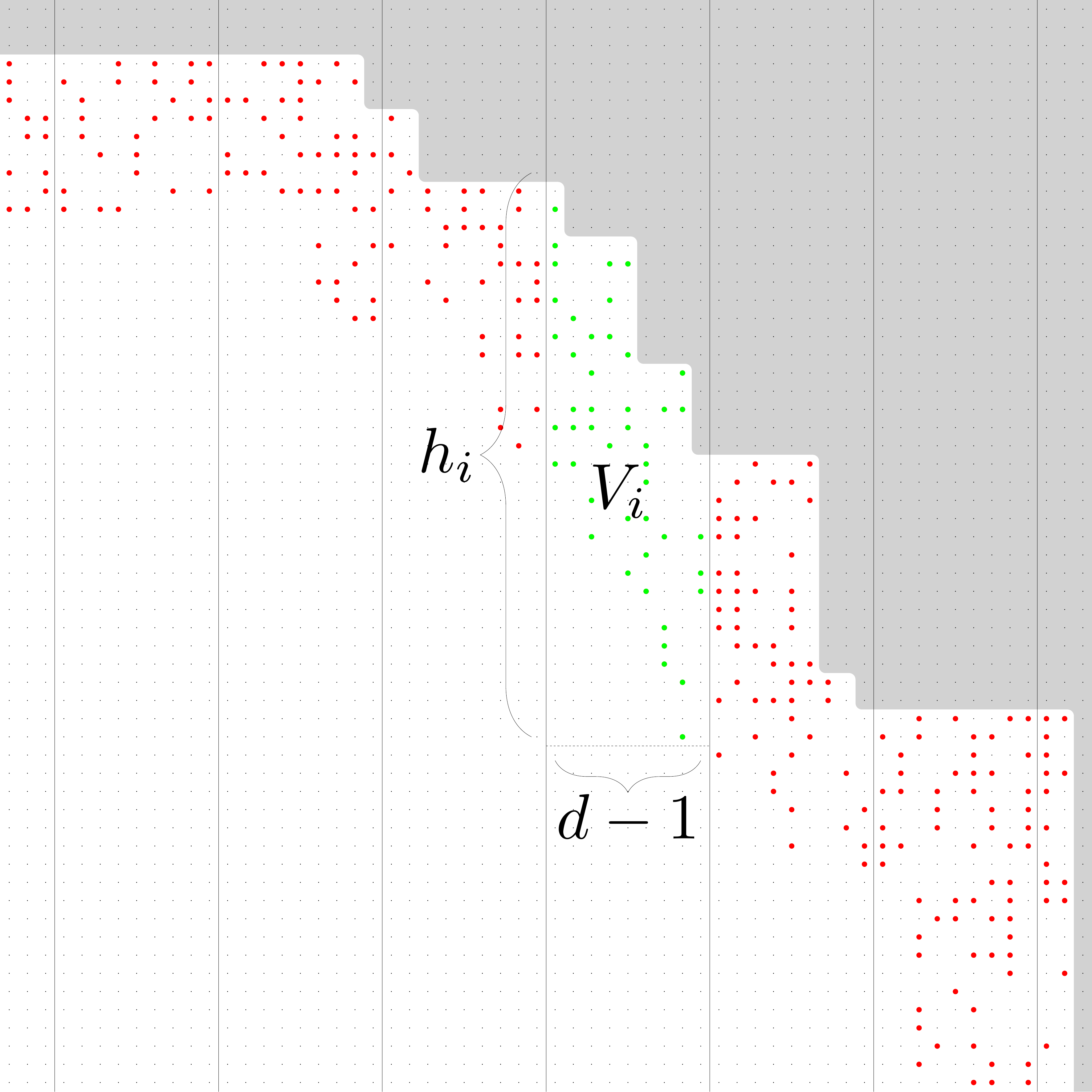}
        \caption{Strings $V_i$ in the partitioning of $F_1$. The thicker points constitute $\d{F_1}$.}
        \label{figure:periphery_decomposition}
    \end{subfigure}
    \caption{Partitioning of $F$. }
\end{figure}

\begin{claim}\label{border_lemma}
Let $F_1$ be a $d$-peripheral string with $\d{F_1} \subseteq K_1$. For all $u \in \d{F_1}$ and ${v \in \d{\Ta}}$, either $\x{v} - \x{u} < d$ or $\y{v} - \y{u} < d$.  
\end{claim}
\begin{proof}
Assume the contrary. Since $u \in \d{F_1}$, there exists $w \in \Z^2 \setminus \d{\Ta}$ such that $\absolute{\x{u} - \x{w}} \le \absolute{u-w} \le d$ and $\absolute{\y{u} - \y{w}} \le \absolute{u-w} \le d$. Since $v \in \d{\Ta}$, there exists $q \in Q$ such that $v \in [m]^2 + q$. We have
\[ \x{w} \ge \x{u} - d \ge n/2 - m/4 \ge n - m > \x{q}\]
and
\[ \x{w} \le \x{u} + d \le \x{v} \le \x{q} + m - 1. \]
Similarly, we can show $\y{q} \le \y{w} \le \y{q} + m - 1$, and thus $w \in [m]^2 + q \subseteq \d{\Ta}$, a contradiction.
\end{proof}

\begin{lemma}\label{lm:partitioning}
A $d$-peripheral string $F_1$ with $\d{F_1} \subseteq K_1$ can be partitioned into $\ell = \O(m/d)$ strings $V_i$ with width at most $d$ and $\sum_{i \in [\ell]} h(V_i) = \O(m)$. (See \cref{figure:periphery_decomposition}.)
\end{lemma}
\begin{proof}
Let $\ell = \lceil n / d \rceil$ and partition $F_1$ into strings $V_i$, $i \in [\ell]$, where $V_i$ is the restriction of $F_1$ to $\set{i \cdot d, \dots, i \cdot d + d - 1} \times [n] \cap \d{F_1}$. (See \cref{figure:periphery_decomposition}.) 
For each $i$, the width of $V_i$ is at most $d$. 

We will denote height bounds by using $h_i = h(V_i)$. We now show that $\sum_{i \in [\ell]} h_i = \O(m)$. For $\ell < 2$ the claim is trivial. Assume $\ell \ge 2$. For every $i \in [\ell]$ (since $h_i$ is minimal) there exists a pair of points $u_i, v_i \in \d{V_i}$ such that $(\x{v_i}, \y{u_i} + h_i - 1) \in \d{\Ta}$. For all $i \ge 2$, we have
\[ h_i \le \y{u_{i - 2}} - \y{u_{i}} + d.\]
since if that were not the case for some $i$, the points $u_{i - 2}$ and $(\x{v_{i}}, \y{u_{i}} + h_i - 1)$ would have contradicted \cref{border_lemma}. We conclude that the sum $h_i$ can be upper-bounded by a telescoping series:
$$\sum_{i = 0}^{\ell - 1} h_i \le h_0 + h_1 + \sum_{i = 2}^{\ell - 1} (\y{u_{i - 2}} - \y{u_i} + d) = \O(m)$$
\end{proof}

As an immediate corollary, we obtain

\begin{corollary}\label{cor:area_F}
$\absolute{\d{F_1}} = \O(dm)$. 
\end{corollary}

Recall the construction of the sets $\V_a$ described in \Cref{sec:partition:pattern}. We call a character $a \in \Sigma$ \emph{frequent} if $\absolute{\V_a} \ge \sqrt{k}$, and otherwise \emph{infrequent}. We partition $F_1$ into two strings $F_1'$ and $F_1''$, where $F_1'$ contains only infrequent characters and $F_1''$ frequent ones. For every $q \in Q$, we then have 

\[\Ham(P + q, S) = \Ham(P + q, F_1') + \Ham(P + q, F_1'').\]

For all $u \in \d{F_1'}$, let $I_u$ be a monochromatic subtile string defined as the restriction of $F_1'$ to $\set{u}$. 
We have $\Ham(P + q, F_1') = \sum_{u \in \d{F_1'}} \Ham(P + q, I_u)$ for every $q \in Q$. 
Consequently, we can compute these values in $\tO(m^2 + \sum_{u \in \d{F'}} \absolute{\V_{\getchar{I_u}}})$ time by \cref{th:sparse_algo}. 
For all $u \in \d{F_1'}$, $\getchar{I_u}$ is an infrequent character, and hence $\absolute{\V_{\getchar{I_u}}} < k^{1/2}$. 
Furthermore, by \cref{cor:area_F}, we have $\absolute{\d{F_1'}} = \O(md)$, and the total complexity of this step is $\tO(m^2 + mdk^{1/2})$. 
To compute the values $\Ham(P + q, F_1'')$, we apply the following lemma:

\begin{restatable}{lemma}{restateLemSigmaBorder}\label{lm:sigma_border}
Let $F_1''$ be a $d$-peripheral string such that $\d{F_1''} \subseteq K_1$. There is an algorithm that computes, for all $q \in Q$, the distance $\Ham(P + q, F_1'')$ in total time $\tO(m^2 + (\sigma+1) \cdot md)$, where $\sigma$ is the total number of distinct characters in $P$ and $F_1''$.
\end{restatable}

Note that $F_1''$ contains $\O(\sqrt{k})$ distinct characters. In $\O(k^2 \log k + m^2) = \tO(m^2)$ time, we replace all characters in $P$ not present in $F_1''$ with a new character. This operation does not change the Hamming distances, and we now have that the total number of distinct characters in $P$ and $F$ is $\O(\sqrt{k})$. We then apply \cref{lm:sigma_border} to compute $\Ham(P + q, F_1'')$ for all $q \in Q$ in $\tO(m^2 + \sqrt k \cdot md)$ time, which completes the proof of \cref{th:dense_algo}. 
\begin{proof}[{Proof of \cref{lm:sigma_border}}]
If $d > m/4$, we can apply \cref{cor:sigman2d} to compute the distances in $\tO(m^2 + (\sigma+1) \cdot m^2) = \tO(m^2 + (\sigma+1) \cdot dm)$ time. From now on, we assume $d \le m / 4$. 

We start by partitioning $P$ into a string $P_0$, defined as the restriction of $P$ to $[m - d]^2$ and a string $P_1$, defined as the restriction of $P$ to $\d{P} \setminus \d{P_0}$.
Further, we partition $P_1$ into a string $P_2$ of width $d$, which is a restriction of $P$ onto $[m-d,m - 1] \times [m - d]$ and a string $P_3$ of height~$d$, which is a restriction of $P$ onto $[m] \times [m-d,m -1]$. (See \cref{figure:pattern_restriction}.) Then, it holds
\[ \Ham(P + q, F_1'') = \Ham(P_0 + q, F_1'') + \Ham(P_2 + q, F_1'') + \Ham(P_3 + q, F_1'').\]
In a moment, we will show that $\Ham(P_0 + q, F_1'') = 0$ for every $q \in Q$. Hence we only have to consider $\Ham(P_2 + q, F_1'')$ and $\Ham(P_3 + q, F_1'')$. We focus on $P_2$, while the proof for $P_3$ is symmetric.

We partition $F_1''$ into $\ell$ strings $V_i$, $i\in [\ell]$ of width at most $d$, where $\ell = \O(m/d)$ by \cref{lm:partitioning}. (In the proof for $P_3$, we have to use a symmetric version of \cref{lm:partitioning} such that the $V_i$'s are of height at most $d$.) Now we show that $\d{P_0+q}$ and $\d{F_1''}$ are disjoint, and thus $\Ham(P_0 + q, F_1'') = 0$.

\begin{claim}\label{border_hamming_reduction}
For all $q \in Q$ and $i \in [\ell]$, it holds $\d{P_0 + q} \cap \d{V_i} = \emptyset$.
\end{claim}
\begin{proof}
Assume towards a contradiction that for some $q \in Q$ there exist $u \in \d{P_0 + q} \cap \d{V_i}$. Define $v = (\x{u} + d, \y{u} + d)$. Since $u \in [m - d]^2 + q$, we have $v \in [m]^2 + q \subseteq \d{\Ta}$. Consequently, we have $u \in \d{V_i}$, $v \in \d{\Ta}$, and $\x{v} - \x{u} = \y{v} - \y{u}$, a contradiction with \Cref{border_lemma}.
\end{proof}

\begin{figure}[!t]
	\begin{center}
		\includegraphics[width=0.8\textwidth]{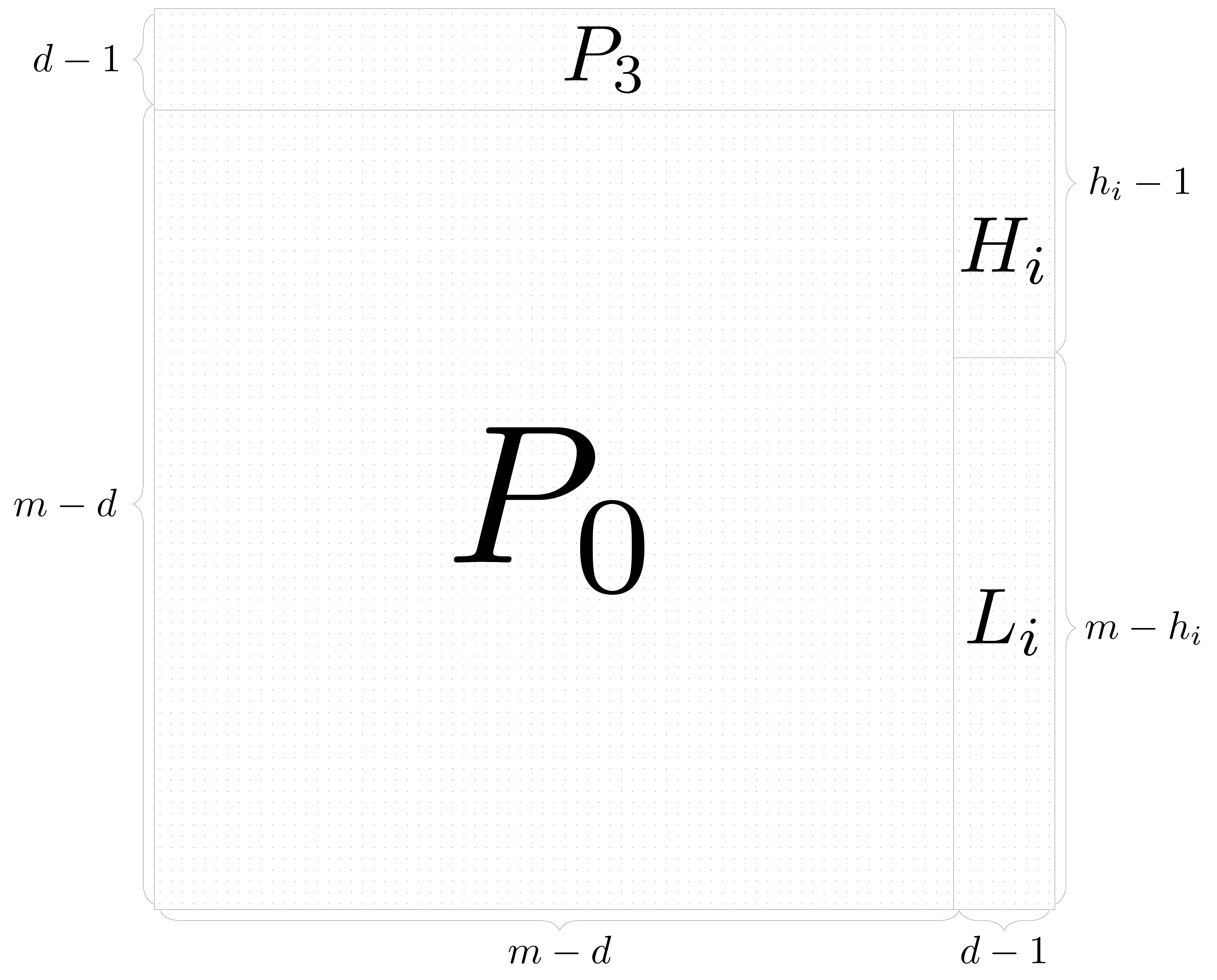}
	\end{center}
	\caption{Pattern partitioning.}
	\label{figure:pattern_restriction}
\end{figure}

Our goal is to compute, for every $q \in Q$, the value

\[ \Ham(P_2 + q, F_1'') = \sum_{i \in [\ell]} \Ham(P_2 + q, V_i).\]

For every $i \in [\ell]$, let $h_i = h(V_i)$. We construct a string $L_i$ as the restriction of $P_2$ to $[m] \times [m - h_i] \cap \d{P_2}$ and a string $H_i$ as the restriction of $P_2$ to $\d{P_2} \setminus \d{L_i}$. (See \Cref{figure:pattern_restriction}.) Since $L_i$ and $H_i$ partition $P_2$, we have

\[\Ham(P_2 + q, V_i) = \Ham(L_i + q, V_i) + \Ham(H_i + q, V_i). \]

\begin{claim}\label{pattern_height_reduction}
For all $q \in Q$ and $i \in [\ell]$, we have $\d{L_i + q} \cap \d{V_i} = \emptyset$ .
\end{claim}
\begin{proof}
Assume towards a contradiction that there exist $q \in Q$ and $i \in [\ell]$ such that $\d{L_i + q} \cap \d{V_i}$ contains $u \in \Z^2$. Define $v = (\x{u}, \y{u} + h_i)$.
Since $u \in [m] \times [m - h_i] + q$, we have $v \in [m]^2 + q \subseteq \d{\Ta}$, thus $v \in \d{\Ta}$, which contradicts the definition of $h_i$.
\end{proof}

By \Cref{pattern_height_reduction}, we have for all $q \in Q$:

\[\Ham(P_2 + q, F_1'') = \sum_{i \in [\ell]} \Ham(P_2 + q, V_i) = \sum_{i \in [\ell]} \Ham(H_i + q, V_i). \]

\paragraph{Wrapping up.}
We apply \cref{cor:sigman2d} to solve the text-to-pattern Hamming distances problem for every pair of $H_i$ and $V_i$. In total, this takes $\tO(\sum_{i \in [\ell]} ((\sigma+1) \cdot d \cdot h_i) = \O((\sigma+1) dm)$ time by \cref{cor:area_F}. We then need to combine these results to obtain $\Ham(P_2 + q, F_1'')$ for all $q \in Q$. 

To do so, we first naively compute, for all $i \in [\ell]$, a set ${q \in \Z^2 : \d{H_i+q} \subseteq \d{V_i}}$ in $\tO(\sum_{i \in [\ell]} d \cdot h_i) = \tO(dm)$ time (again, by \cref{cor:area_F}). Note that the union of these sets contains all vectors $q \in Q$ for which the value $\Ham(P_2 + q, F_1'')$ can be different from zero. 

We then compute $Q' = Q \cap \bigl( \cup_i \{q \in \Z^2 : \d{H_i+q} \subseteq \d{V_i}\} \bigr)$ by sorting and merging the two sets in $\tO(m^2 + dm) = \tO(m^2)$ time. Note that $\absolute{Q'} = \O(dm)$ and $Q'$ does not contain duplicates. For each $q \in Q'$, we compute $\Ham(P + q, F_1) = \sum_{i \in [\ell]} \Ham(H_i + q, V_i)$ in $\tO(\sum_{i \in [\ell]} d \cdot h_i) = \tO(dm)$ total time, completing the proof of \cref{lm:sigma_border}.
\end{proof}

\bibliographystyle{plainurl}
\bibliography{references}
\appendix
\section[Constructing a parallelogram grid]{Proof of \cref{lm:parallelogram_grid}}
\label{app:parallelogram_grid}

\restateLemParallelogramGrid*

\begin{proof}
\def\paragraphrp#1{\vspace{.5\baselineskip}

\noindent\textbf{#1}}

\paragraphrp{Construction of the grid and Property (\ref{it:unique}).}
Recall that the grid is defined by a sequence of $\ell+1$ lines $h_0, h_1, \ldots, h_\ell$ parallel to $\phi$ and a sequence of $\ell+1$ lines $v_0, v_1, \ldots, v_\ell$ parallel to $\psi$. 

To construct the lines $h_0, h_1, \ldots, h_\ell$, we first compute $\h{u}$ for all $u \in [n^2]$ in $\O(n^2) = \O(m^2)$ time. As $\phi \in (0,+\infty) \times [0,+\infty]$, we have 
$$-\absolute{\phi} \cdot n \le \h{u} = \x{u} \cdot \y{\phi} - \y{u} \cdot \x{\phi} \le \absolute{\phi} \cdot n.$$
We then sort these values in $\O(n^2 \log n) = \O(m^2 \log m)$ time and delete duplicates. Let $-\absolute{\phi} \cdot n \le c_1 < \ldots < c_{N} \le \absolute{\phi} \cdot n$, where $1 \le N \le n^2$, be the resulting values. Let $\delta = (\min_{i=1}^{N-1} (c_{i+1}-c_i))/2\ell = \O(n/\ell)$. Define $\alpha_0 = c_1 -\delta$ and $\alpha_\ell = c_{N} + \delta$. Next, choose equispaced values $\alpha_1, \alpha_2, \ldots, \alpha_{\ell-1}$ (necessarily rational) in the interval $[\alpha_0, \alpha_{\ell}]$. If for some $i, j$ we have $\alpha_i = c_j$, set $\alpha_i = \alpha_i + \delta$. We now have that for all $i$, each of the following is satisfied:
\begin{enumerate}
\item $\alpha_i \notin \{\h{u} : u \in [n^2]\}$;
\item $\alpha_{i+1}-\alpha_i = \O(\absolute{\phi} n / \ell)$;
\item For all $u \in [n^2]$, there is $\alpha_0 < \h{u} < \alpha_\ell$.
\end{enumerate}
Finally, we define $h_i = \{u: u \in \R^2, \h{u} = \alpha_i\}$. In total construction of the lines $h_i$ takes $\O(m^2 \log m)$ time.

We construct the lines $v_0, v_1, \ldots, v_\ell$ analogously. Namely, we select in $\O(m^2 \log m)$ time rational values $\beta_i$, $0 \le i \le \ell+1$, satisfying each of the following:
\begin{enumerate}
\item $\beta_i \notin \{\s{u} : u \in [n^2]\}$;
\item $\beta_{i+1}-\beta_i = \O(\absolute{\psi} n / \ell)$;
\item For all $u \in [n^2]$, there is $\beta_0 < \s{u} < \beta_\ell$.
\end{enumerate}
We then define $v_i = \{u: u \in \R^2, \s{u} = \beta_i\}$. 

For $0 \le i,j \le \ell+1$ let $w_{i,j}$ be the intersection of $h_i$ and $v_j$ (defined correctly as they are not collinear). For $i,j \in [\ell]$ define a parallelogram $p_{i,j}$ as the parallelogram with vertices $w_{i-1,j-1}$, $w_{i-1,j}$, $w_{i,j-1}$, $w_{i,j}$. Property $\ref{it:unique}$ easily follows from the properties of $\{\alpha_i\}$ and $\{\beta_j\}$.

\paragraphrp{Property (\ref{it:small_parallelogram}).}
Consider any $p_{i,j}$, $i,j \in [\ell]$. First note that	for all $u, v \in p_{i, j}$, we have $\absolute{u - v} = \O(n / \ell)$. Consider any $u, v \in p_{i, j}$ and denote $w = u - v$.
By definition of $p_{i, j}$, we have
\[ \absolute{\h{w}} = \absolute{\h{(u - v)}} = \absolute{\h{u} - \h{v}} = \O(n\absolute{\phi} / \ell) \]
and similarly $\absolute{\s{w}} = \O(n\absolute{\psi} / \ell)$. Since $\phi$ and $\psi$ are not collinear, there exist $s, t \in \mathbb{R}$, such that $w = s\phi + t\psi$. Recall that by \Cref{get_periods} we have $\absolute{\phi \times \psi} \ge \frac{1}{2}\absolute{\phi}\absolute{\psi}$ (since $\absolute{\sin \alpha} \ge 1/2$), thus
\[ \frac{1}{2}\absolute{t}\absolute{\phi}\absolute{\psi} \le \absolute{t}\absolute{\phi \times \psi} = \absolute{\phi \times (s\phi + t\psi)} = \absolute{\h{w}} = \O(n\absolute{\phi} / \ell), \]
which gives us $\absolute{t\psi} = \O(n / \ell)$. We can similarly prove that $\absolute{s\phi} = \O(n / \ell)$ and finally 
\[ \absolute{w} = \absolute{s\phi + t\psi} \le \absolute{s\phi} + \absolute{t\psi} = \O(n / \ell). \]

Secondly, we show that	for all $u \in \X(p_{i, j}) \times \Y(p_{i, j})$ there exists $v \in p_{i, j}$, such that $\absolute{u - v} = \O(n / \ell)$. Indeed, there exist $w \in p_{i, j}$ such that $\x{u} = \x{w}$ and $v \in p_{i, j}$ such that $\y{u} = \y{v}$. From above, 
\[\absolute{u - v} = \absolute{\x{u} - \x{v}} = \absolute{\x{w} - \x{v}} \le \absolute{w - v} = \O(n / \ell).\] 

Finally, consider any $u, v \in p_{i, j}$. From above, there exist $u', v' \in p_{i, j}$ such that $\absolute{u - u'} = \O(n / \ell)$ and $\absolute{v - v'} = \O(n / \ell)$. Therefore, 
\[ \absolute{u - v} \le \absolute{u - u'} + \absolute{u' - v'} + \absolute{v' - v} = \O(n / \ell). \]

\paragraphrp{Property (\ref{it:monotonicity}).} 
We show that for all $i \in [\ell - 1]$ and $j \in [\ell]$ we have
$$\min \X(p_{i, j}) < \min \X(p_{i + 1, j}),$$
the rest can be shown analogously. Recall that $w_{s,t}$ is the intersection of the lines $h_s$ and $v_t$ defining the parallelogram grid, $\phi \in [0, +\infty] \times (-\infty,0)$, and $\psi \in (0,+\infty) \times [0,+\infty)$. 

As $p_{i,j}$ and $p_{i+1,j}$ are parallelograms, we have 
\begin{align*}
\min \X(p_{i, j}) = \min\{\x{w_{i, j}}, \x{w_{i + 1,j}}, \x{w_{i + 1,j + 1}},\x{w_{i,j + 1}}\} = \x{w_{i,j + 1}},\\ 
\min \X(p_{i+1, j}) = \min\{\x{w_{i + 1, j}}, \x{w_{i + 2, j}}, \x{w_{i + 2,j + 1}},\x{w_{i+1,j + 1}}\} = \x{w_{i + 1,j + 1}}.
\end{align*}

From the construction of $h_s$ and $v_t$ it follows that $w_{i + 1,j + 1} = w_{i,j+1} + \beta \cdot \psi$ for some $\beta > 0$, and hence $\x{w_{i + 1,j + 1}} > \x{w_{i,j + 1}}$. Therefore, $\min \X(p_{i+1, j}) > \min \X(p_{i, j})$ as desired. 
\end{proof}
\end{document}